\documentclass[journal]{IEEEtran}

\usepackage{amssymb,amsmath,amsfonts,amsthm}
\usepackage{algorithm,algorithmic}
\usepackage{boxedminipage}
\usepackage{cite}
\usepackage[dvips]{graphicx}
\usepackage{epsfig}
\usepackage{float}
\graphicspath{{figures/}}
\usepackage{graphicx}
\usepackage{graphics}
\usepackage{tabulary}
\usepackage{booktabs}

\usepackage{multirow}
\usepackage{threeparttable} 
\usepackage[usenames,dvipsnames]{color}
\usepackage{url,bm,xspace,dsfont}

\usepackage{verbatim}

\usepackage{calrsfs}

\newtheorem{theorem}{Theorem}[section]
\newtheorem{lemma}[theorem]{Lemma}
\newtheorem{definition}[theorem]{Definition}

\newtheorem{problem}[theorem]{Problem}
\newtheorem{example}[theorem]{Example}
\newtheorem{remark}[theorem]{Remark}

\newtheorem{assumptions}[theorem]{Assumptions}

\DeclareMathAlphabet{\pazocal}{OMS}{zplm}{m}{n}

\renewcommand{\theequation}{\thesection.\arabic{equation}}
\renewcommand{\thefigure}{\thesection.\arabic{figure}}
\renewcommand{\thetable}{\thesection.\arabic{table}}

\newcommand{\sr}{\stackrel}

\newcommand{\rar}{\rightarrow}

\newcommand{\tri}{\sr{\triangle}{=}}

\newcommand{\bea}{\begin{eqnarray}}
\newcommand{\eea}{\end{eqnarray}}
\newcommand{\bes}{\begin{eqnarray*}}
\newcommand{\ees}{\end{eqnarray*}}
\newcommand{\bce}{\begin{center}}
\newcommand{\ece}{\end{center}}

\def\VR{\kern-\arraycolsep\strut\vrule &\kern-\arraycolsep}
\def\vr{\kern-\arraycolsep & \kern-\arraycolsep}

\newcommand{\ben}{\begin{enumerate}}
\newcommand{\een}{\end{enumerate}}

\newcommand{\bi}{\begin{itemize}}
\newcommand{\ei}{\end{itemize}}

\newcommand{\bp}{\begin{problem}}
\newcommand{\ep}{\end{problem}}
\newcommand{\hso}{\hspace{.1in}}
\newcommand{\hst}{\hspace{.2in}}

\newcommand{\noi}{\noindent}

\newcommand{\bc}{\begin{center}}
\newcommand{\ec}{\end{center}}

\begin{document}
\title{ Private \& Common   Information States in Decentralized    Team Equilibrium  via Dynamic Programming for  POMDPs
  with Delayed Sharing
}

\author{Charalambos D. Charalambous$^1$, Umarbek Guvercin$^2$,  Seddik   Djouadi$^2$ 
\thanks{$^{1}$Charalambos D. Charalambous is  with the Faculty of the Department of Electrical and Computer Engineering, University of Cyprus, Nicosia 1678, Cyprus
{\tt\small chadcha@ucy.ac.cy}}
\thanks{$^{2}$Seddik Djouadi and  Umarbek Guvercin are  with the Faculty of the Department  Electrical Engineering and Computer Science, University of Tennessee, Knoxville, TN, 37996, USA 
{\tt\small \{mdjouadi,ugvercin\}@utk.edu}}%
\thanks{}%
}

%

\maketitle

\begin{abstract}
Witsenhausen in his seminal 1971 paper \cite{witsenhausen1971}, introduced  decentralized partially observable Markov decision problems (POMDPs),  with  multiple agents or controls  operating under   $T-$step delayed sharing information patterns. A fundamental problem in \cite{witsenhausen1971}
is  the identification of structural properties of optimal strategies that compress the  information patterns into multiple   information states.



In this paper,  we develop such structural properties of optimal strategies and associated dynamic programming (DP) equations,   using  the  concept of decentralized sequential team equilibrium (a generalization of  person-by-person  optimality from static team theory). Within this framework, 
each strategy  is assigned an individual value function conditioned on its delayed sharing information pattern, while the strategies  of all other agents  are held fixed.
The resulting   DP framework  yields several new  DP equations and  characterizations of   decentralized team equilibrium. Moreover, these   DP equations exhibit fundamental properties analogous to those of centralized DP  of POMDPs:  the
optimization in each   agent's  DP equations is performed   over the  agent's action space rather than  over   strategy spaces;  each agent's  multiple information states satisfy   Markov recursions;
  and a separation principle   holds. 

 The DP equations 
  reveal  a   structural compression property  of   optimal strategies:  each agent compresses its 
delayed sharing  information pattern  into three  components:   1) a private   posterior distribution conditioned  on the agent's delayed sharing  information pattern, 
  2)  a  centralized  posterior    distribution conditioned  on  the common information  shared by    all agents, and 3) the agent's private information component. This structural result   essentially     completes  Witsenhausen's Assertion 8 in  \cite{witsenhausen1971}.

%

\end{abstract}


\section{INTRODUCTION}
\label{sec:int}
Witsenhausen in the 1971  seminal   paper \cite{witsenhausen1971},    stimulated  vast    interest in decentralized  stochastic control and decision systems,   with multiple decision makers, agents or controls, each acting on its own information pattern or structure to  affect the overall  system performance.
 The formulation in  \cite{witsenhausen1971} is focused  on   Markovian decentralized  stochastic networks  with    multiple  observation posts collecting information and  multiple control stations  applying control actions.
 In \cite{witsenhausen1971}, several   challenges of decentralized optimization are discussed and  research directions are proposed, with  emphasis on the principle of separation of estimation and control for problems with   {\it  delayed sharing information patterns}\cite[Section~IV.G]{witsenhausen1971}.

Over the last 50 years,  several decentralized optimizations methods were  proposed.
  Most of these methods were   inspired  by concepts of   {\it static team theory} 
\cite{marschak1955,radner1962,marschak-radner1972}
called   {\it person-by-person\footnote{PbP optimality reminiscent to  the concept  of Nash-equilibrium; it is applied   in static  optimization problems with multiple  agents  having access to asymmetry of information, who  aim  to optimize a common payoff.} (PbP) optimality,  team  optimality, and   their connection via  Radner's  theorem}  \cite{radner1962,marschak-radner1972}, which  identified  sufficient 
conditions   for   team optimality   in terms of  stationary conditions of  PbP optimality.

In the early 1970's, several studies \cite{kurtaran-sivan1973,sandell-athans1974,kurtaran1975,yoshikawa1975,varaiya-walrand1978,kurtaran1979,hsu-marcus:1982,aicardi-davoli-minciardi1987} focused on  partially observable Markov decision problems (POMDPs), when the multiple controls have access to $1-$step delayed sharing information patterns, making use of  Witsenhausen's separation  \cite[Assertion 8, pp.1562]{witsenhausen1971}. Studies  \cite{sandell-athans1974,yoshikawa1975,varaiya-walrand1978,hsu-marcus:1982,aicardi-davoli-minciardi1987}   developed  a single DP equation that corresponds to {\it a single cost-to-go} for all strategies  conditioned only on  the shared  information component available to all strategies. However, Varaiya's and Walrand's \cite{varaiya-walrand1978} counterexample to  \cite[Assertion~8, pp.1562]{witsenhausen1971} for  $2-$step delayed sharing patterns, had a discouraging effect  on  further progress, which lasted for several years. More recently,   the     single DP equation is further  developed  in  \cite{nayyar-mahajan-teneketzis2011,nayyar-mahajan-teneketzis2013,nayyar-teneketzis2019,TNTT2023},   for  arbitrary $T-$step delayed sharing patterns,  using variations of \cite[Assertion 8, pp.1562]{witsenhausen1971}.

%

This paper  develops  an alternative DP framework, based on  the concept of decentralized sequential team equilibrium (a generalization of the concept of PbP optimality). We show that this alternative  DP framework   leads  to
\begin{enumerate}
\item several new distinct DP equations that characterize  decentralized sequential team equilibrium, 

\item structural compression properties of the optimal strategies, utilizing multiple information states or beliefs, and 

\item  the DP equations and information states retain the fundamental properties of classical centralized POMDPs, and a generalized separation principle between estimation and control holds. 
\end{enumerate}
Essentially, we prove that Witsenhausen's compression structural property of the optimal strategies  \cite[Assertion~8]{witsenhausen1971}, requires an additional  
information state, corresponding to  the private   posterior distribution of  the agent's    unobservable  states  conditioned on its own  delayed sharing  information pattern. 

Although, this paper is focused on    decentralized sequential team equilibrium, we also show 
  that  the structural compression property of optimal strategies  also holds for the  single cost-to-go for all strategies  conditioned only on  the shared  information component available to all strategies \cite{sandell-athans1974,yoshikawa1975,varaiya-walrand1978} and  \cite{nayyar-mahajan-teneketzis2011,nayyar-mahajan-teneketzis2013,nayyar-teneketzis2019}. 


In the rest of this section we describe the new DP approach and the main  contributions of this paper (Section~\ref{sub-DP}). In addition,   we  provide a brief summary of the main   decentralized optimization methods,  their   advantages and disadvantage (Section~\ref{sect:lit}).

 \subsection{The New Decentralized DP Approach and Separation}
 \label{sub-DP} 
\subsubsection{Main Problem (Section~\ref{sect:model})} 
We consider  
partially observable Markov decision problems (POMDPs),  with multiple  agents, controls or decision makers. Each agent has   access to a $T-$step delayed  sharing information pattern \cite[Section~IV.G]{witsenhausen1971}, that  
  consists of a  private information component and a $T-$step shared or common information component received by all other agents with delay $T$ unit of time. 

We formulate the decentralized optimization  of POMDPs using  the notion of a decentralized (sequential) team equilibrium, called PbP optimality in  static team theory \cite{marschak1955,radner1962,marschak-radner1972}. 
 Each agent's strategy optimizes a common (team)  payoff, when all other strategies are fixed to their optimal strategies  (see Definition~\ref{def-pbp}).
    Based on the definition of  decentralized team equilibrium, for each agent we introduce a   cost-to-go   conditioned on the agent's $T-$step delayed sharing information pattern, when all other agents' strategies are fixed to their optimal strategies.
    
 The definition of   decentralized team equilibrium  and associated cost-to-go  is fundamentally different from past studies \cite{kurtaran-sivan1973,sandell-athans1974,kurtaran1975,yoshikawa1975,varaiya-walrand1978,kurtaran1979,hsu-marcus:1982,aicardi-davoli-minciardi1987,nayyar-mahajan-teneketzis2011,nayyar-mahajan-teneketzis2013,nayyar-teneketzis2019}. Specifically, \cite{kurtaran-sivan1973,sandell-athans1974,kurtaran1975,yoshikawa1975,varaiya-walrand1978,kurtaran1979,hsu-marcus:1982,aicardi-davoli-minciardi1987,nayyar-mahajan-teneketzis2011,nayyar-mahajan-teneketzis2013,nayyar-teneketzis2019} considered  a single  cost-to-go for  all strategies, conditioned  only on the common or shared information component of  all strategies (see  Section~\ref{sec:lr}), leading to a single DP equation, while our formulation is based on multiple DP equations, one for each agent. 

\subsubsection{Contributions}  We propose a  new generalized  DP and separation principle framework,  which  retains  the   two fundamental  properties of classical  centralized  POMDPs:
 
{\bf Property \#1.} Each agent's  cost-to-go is a functional of  the agent's $T-$step delayed sharing  information pattern. 
 The cost-to-go of each agent satisfies a DP equation in which the optimization   is over the agent's actions and not their   strategies (all other strategies fixed to their optimal strategies). 

{\bf Property \#2.} Each agent's  DP equation is  expressed in terms of  1)  the agent's private  posterior distribution,  which satisfies a  conditionally  Markov recursion,  2) a   common posterior distribution available to  all agents, which  also satisfies a Markov recursion, and 3) a limited length private information component, all depending  on agent's  actions and not its  strategies.   The private  and common posterior distributions and  the  limited length private information component  constitute an   information state
 and  a  sufficient statistic for  the  agent's strategies,  and   a generalized separation principle holds.

 Properties \#1, \#2 are generalizations of  the fundamental properties  of  classical  centralized DP  of  POMDP  \cite{bertsekas-shreve1978,kumar-varayia:B1986,ahmed1998,paris-borkar-emmanuel-ghoshi-marcus93,hernandezlerma-lasserre1996,striebel1965}.
In particular, it is well-known that  Property \#1 is essential to  develop efficient DP algorithms, while Property \#2 leads to  Markovian   strategies, rather than strategies which  depend on  data that expand with time. 

%

%
%
%

We present several distinct DP equations   with different choices  of  the agent's information states, and  we prove  the  following   main results.  
\bi
\item {\it  Necessary and sufficient conditions for  decentralized sequential  team equilibrium  and separation principle},   using generalized   DP equations and information states that exhibit  
Properties \#1, \#2.  
\ei

{\bf Compression of Agent's Data.} The  generalized DP equations illustrate the role of the private and the common  information components of each agent's strategy to  compress the agent's  $T-$step delayed sharing information pattern.  This  gives rise to  agent's strategies with   different information states or beliefs, 
	 called, {\it semi-separated strategies},   {\it separated strategies}, and {\it information state strategies} (as we explain below), which are  natural generalizations of  analogous information states and separated strategies of the    classical centralized DP  approach  of POMDP
	  \cite{bertsekas-shreve1978,kumar-varayia:B1986,ahmed1998,paris-borkar-emmanuel-ghoshi-marcus93,hernandezlerma-lasserre1996,striebel1965}.

We prove the  important {\it structural compression  property} that  {\it separated strategies are optimal:}     each agent   compresses the  $T-$step delayed sharing information pattern into three components.

1) {\it The  agent's private information state:} a joint  conditional probability distribution  of the   unobservable state of the network and  all other agents private information components not observed by the  agent,  conditional on the agent's information pattern.

2) {\it  The    centralized information state:} a common  conditional probability distribution for  all agents based  only on the  shared (common) information component available to  all agents.

3) {\it  The agents'  private information component w.r.t. all other agents:} the agent's private  information component not observed by  all other agents. 

Our separated strategies with components 1)-3)  is an enlargement  of 
Witsenhausen's  strategies  in  \cite[Assertion 8, pp.1562]{witsenhausen1971} that include only components 2) and 3).

We also prove an additional  {\it structural property}, using a different choice of the centralized information state.

For completeness,  we briefly discuss (Section~\ref{sec:lr})   applications of  our compression structural property of   optimal strategies to   past  DP studies   \cite{kurtaran-sivan1973,sandell-athans1974,kurtaran1975,yoshikawa1975,varaiya-walrand1978,kurtaran1979,bagchi-basar1980,hsu-marcus:1982,aicardi-davoli-minciardi1987,yuksel-basar:B2013,nayyar-mahajan-teneketzis2011,nayyar-mahajan-teneketzis2013,nayyar-teneketzis2019,TNTT2023}, which are  based  on  {\it a single cost-to-go} for all strategies conditioned on the common information component of all agents, and   make use  Witsenhausen's   \cite[Assertion 8, pp.1562]{witsenhausen1971} or  its variants.

\subsection{Literature  Review}
\label{sect:lit}

The 
 progress in 
 decentralized optimization of  stochastic  dynamical systems may be  classified  into  the following   3 main directions,

{\it  1) Dynamic Programming  for  Delayed  Sharing Patterns. } 

 {\it  2) Static Reduction  for Arbitrary Information Patterns.}
 
 {\it  3)  Pontryagin's Decentralized  Stochastic Maximum Principle (SMP) for Arbitrary Information Patterns.}


\subsubsection{Dynamic Programming  Approach Based on Common Information \cite[Assertion 8, pp.1562]{witsenhausen1971}.}  Earlier   studies   \cite{kurtaran-sivan1973,sandell-athans1974,kurtaran1975,yoshikawa1975,varaiya-walrand1978,kurtaran1979,hsu-marcus:1982,aicardi-davoli-minciardi1987} considered $1-$step delayed sharing patterns, by invoking    Witsenhausen's separation     \cite[Assertion 8, pp.1562]{witsenhausen1971}.  Studies  \cite{sandell-athans1974,yoshikawa1975,varaiya-walrand1978,hsu-marcus:1982,aicardi-davoli-minciardi1987}  developed  a  single DP equation     based on  {\it a single  cost-to-go} for all strategies conditioned on the shared or common information   of all strategies    (see also Section~\ref{sec:lr}). 
The optimal strategies of linear-quadratic Gaussian (LQG) problems with  $1-$step delayed sharing patterns  are derived in  \cite{kurtaran-sivan1973,sandell-athans1974,kurtaran1975,yoshikawa1975}.\\
\indent 
In spite of  the success for $1-$step delayed sharing patterns,  Varaiya and Walrand   in  1978,  \cite[Section~II]{varaiya-walrand1978} proved  via a counterexample,  that  the  single DP equation based on  \cite[Assertion 8, pp.1562]{witsenhausen1971}, does not extend to $2-$step delayed sharing patterns.
\indent 
In view of the counterexample \cite{varaiya-walrand1978}, the single DP approach based on    variations of Witsenhausen's  \cite[Assertion 8, pp.1562]{witsenhausen1971} is  re-visited  in \cite{kurtaran1979} and more recently in    \cite{nayyar-mahajan-teneketzis2011,nayyar-mahajan-teneketzis2013,nayyar-teneketzis2019,TNTT2023}. However,  it is acknowledged  in   \cite{nayyar-mahajan-teneketzis2011} that the information structures of the optimal strategies  cannot be derived from those  in \cite{kurtaran1979}, and vice-versa.

A fundamental question that remained open since the appearance of    \cite[Assertion 8, pp.1562]{witsenhausen1971}, is to determine how each agent compresses its assigned $T-$step delayed sharing information pattern.
 This paper contributes in   this direction.

\indent  Below,  we identify    several silent  technical limitations of  the single DP  approach, which motivated  our investigation to develop a DP approach based on decentralized team equilibrium (i.e. PbP optimality).

 
 
The single DP approach presupposes the existence of  an   information sharing component  among all the agents.  If  there is no   information sharing among the agents,  then  the conditioning information in the   {\it single  cost-to-go}   is the {\it null set}, and hence no DP equation  can be derived.  This technical limitation  is  acknowledged  by Witsenhausen in  \cite[below Assertion 8, pp.1562]{witsenhausen1971} and also in  \cite[Abstract]{nayyar-teneketzis2019}.\\
\indent  Perhaps the most important disadvantage of the single  DP  approach is  the fact that  it  does not fully  characterize the optimality of the strategies;  this is obvious  from Sandell's and  Athan's  \cite{sandell-athans1974} derivation  
of the solution of  LQG problems with $1-$step delayed sharing patterns (and    \cite{kurtaran-sivan1973,kurtaran1975,yoshikawa1975}). Indeed,    
the single DP equation is an   intermediate or   preliminary  step, towards obtaining the optimal strategies.   The complete derivation   involves  another decentralized   optimization problem, which is solved  using  existing results of PbP optimality of static team theory \cite{radner1962}, and the simplicity of  $1-$step delayed sharing. Similar steps are  involved in the
 counterexample
   \cite{varaiya-walrand1978}. 

Apparent from \cite{kurtaran1979,nayyar-mahajan-teneketzis2011,nayyar-mahajan-teneketzis2013,nayyar-teneketzis2019}, is  that  in general the {\it single DP  approach}  based on the common information,   does  not retain the  fundamental Properties  of classical centralized DP of POMDPs. 
In particular,   the single   DP equation involves an optimization 
   over the strategy spaces and not the action spaces,  and the single  information state based on variants of  \cite[Assertion 8, pp.1562]{witsenhausen1971}  depends on the strategies of the agents and not on their actions. These technical issues are also acknowledged in    \cite[pp.1610-1611]{nayyar-mahajan-teneketzis2011}.  

In this paper, we   develop an alternative DP  approach based on decentralized team equilibrium, i.e., PbP optimality.   It is  
motivated by the fact that, 
 for the simplest   {\it static team problems} with $K$ decision makers \cite{radner1962,marschak-radner1972},     it is impossible to characterized     directly team  optimality  by  a {\it single stationary condition}.  On the contrary, first one characterizes    PbP optimality by  $K$    stationary conditions,  one for each decision maker, while under additional convexity condition, Radner's theorem \cite{radner1962}  establishes that PbP optimal strategies are  also team optimal.  In the context of decentralized optimization of stochastic dynamical systems, this  implies  we should aim at developing a DP approach, based on      the cost-to-go of each agent's strategy conditional on the agent's information pattern. Essentially, our  approach is  analogous to the DP approach  of   non-cooperative games  of control, that makes use of the   notion of Nash-equilibrium strategies.  
It leads  to a  new  DP approach and separation principle, based on multiple DP equations  that satisfy  Properties \#1, \#2.
 
%
%

 \subsubsection{Equivalence of Decentralized Stochastic Dynamical  Systems 
 and  Static Teams  \cite{witsenhausen1988}.} Another major development   is    the  {\it static reduction approach} proposed by  Witsenhausen in  1988    \cite{witsenhausen1988} and further developed in \cite{charalambous2014equivalence} for arbitrary non-Markovian decentralized stochastic systems  and   information patterns, using Girsanov's   change of measure \cite{liptser-shiryayev1977}. In particular, Girsanov's   change of measure \cite{liptser-shiryayev1977} is applied  to construct a Radon-Nikodym derivative (RND), such that  on  a reference   probability space  the observations and states are independent of the control strategies.  Effectively transforming the original decentralized stochastic  dynamical  problem to the   equivalent static reduction problem, so that static team theory results \cite{radner1962,marschak-radner1972} are  directly applicable.

  
  
However,     in spite  the   fact that necessary PbP stationary conditions can be  derived for the static reduction problem \cite{charalambous2014equivalence}, this does not mean these  conditions are  sufficient for PbP optimality.
   Moreover, to establish  team optimality using a generalization of    Radner's theorem   \cite{radner1962,marschak-radner1972} is very challenging. The reason is, sufficiency and  Radner's theorem  require   convexity of the payoff of the static reduction problem on the  reference measure, which is rarely convex.  For example, the payoff of  Witsenhausen's counterexample \cite{witsenhausen1968} is convex in the two strategies but the equivalent payoff of the static  reduction on the  reference measure is not  convex \cite{charalambous2014equivalence,teslang-djouadi-charalambous:ACC2021}. 
 Nevertheless, the  method is applied in \cite{teslang-djouadi-charalambous:ACC2021}   to  the Witsenhausen counterexample to  show the two optimal strategies satisfy two coupled  integral equations,  and  to prove  a fixed point theorem,  establishing existence and uniqueness of solutions, essentially solving the counterexample.

 \subsubsection{Pontryagin's Decentralized  Stochastic Maximum Principle (SMP)-Strong and Weak Formulations
 \cite{charalambous-ahmed:IEEEAC2017a,charalambous-ahmed:IEEEAC2018,charalambous-ahmed:IEEEAC2017b,charalambous:MCSS2016}.} Another approach is   Pontryagin's decentralized stochastic maximum principle (SMP)  \cite{charalambous-ahmed:IEEEAC2017a,charalambous-ahmed:IEEEAC2018,charalambous-ahmed:IEEEAC2017b,charalambous:MCSS2016}. This approach is a natural  generalization of     static team theory and 
 Radner's theorem of stationary conditions  to decentralized stochastic systems described  by  stochastic differential equations (SDEs) with arbitrary information patterns,  and strategies which may not have perfect recall \cite{witsenhausen1971}. Using  the  decentralized SMP various necessary and sufficient conditions for PbP and team optimality are obtained  \cite{charalambous-ahmed:IEEEAC2017a,charalambous-ahmed:IEEEAC2018,charalambous-ahmed:IEEEAC2017b,charalambous:MCSS2016}, which are expressed in  terms of conditional  Hamiltonian systems and backward  SDEs. The necessary and sufficient conditions take different forms, depending on whether these are derived using 
  the strong or  the weak  (Girsanov's change of measure) formulations of the underlying probability spaces  \cite{yong-zhou1999,elliott-aggoun-moore1995}.  
  These include various convexity conditions  such that PbP optimality implies team optimality (i.e., equivalent to Radner's theorem). 
 
However, even though   the decentralized SMP is a natural generalization  of static team theory to a dynamic team theory, 
  the decentralized SMP   is 
  computationally  demanding,    because it uses forward and backward SDEs. 
  
  Alternative methods are described in \cite{waal-vanschuppen2000,huang2012social,xiangyu-kaiqing:2026}. 
 The rest of the paper is structured as follows.

Section~\ref{sect:model} introduces  the  decentralized stochastic  network.

Section~\ref{discrete}   derives the various DP equations  and the new structural compression property of optimal strategies.

Section~\ref{sec:lr} discusses   applications of  our new structural compression property of optimal strategies  to past DP approaches   based  on   a single cost-to-go.

Section~\ref{sect:ex}, illustrates  the implementation of one of the   DP equations to a specific example.

\section{The  Decentralized Stochastic  Network}
\label{sect:model}
Here, we introduce 
 the decentralized stochastic network.

\subsection{Notation}
\label{sec:nota}
$\mathbb{R} \tri (-\infty,\infty)$, $\mathbb{Z} \tri \{\ldots, -1,0,1, \ldots\}$,    $\mathbb{Z}_+ \tri \{1,2, \ldots\}$, {$\mathbb{Z}_+^n \tri \{1,2, \ldots,n\}$,}  $n \in {\mathbb Z}_+$. 
  Given  $s^{(K)} \tri \{s^1, \ldots, s^K\}, K \in \mathbb{Z}_+$, we  define   $s^{-k} \tri s^{(K)} \setminus\{s^k\}=\{s^1, \ldots, s^{k-1},s^{k+1}, \ldots, s^K\}$, i.e.,  $s^{(K)}$ minus element $\{s^k\}$.  \\ 
 $\{({\mathbb X}_t,{\cal B}({\mathbb X }_t))\big| t\in\mathbb{Z}_+^n\}$ denotes  measurable spaces, where ${\mathbb X}_t$ is  a complete separable metric space, known as Polish  or Borel space, and ${\cal B}({\mathbb X}_t)$ is the Borel $\sigma-$algebra of subsets of ${\mathbb X}_t,  \forall t\in\mathbb{Z}_+^n$. Points in the product space ${\mathbb X}_{1,n}\triangleq{{\prod}_{t\in\mathbb{Z}_+^n}}{\mathbb X}_t$ are denoted by $x_{1,n}\triangleq \{x_1,\ldots, x_n\} \in{\mathbb X}_{1,n}$, and their restrictions for any $(m,n)\in\mathbb{Z}_+ \times \mathbb{Z}_+$ by $x_{m,n}\triangleq \{x_{m},\ldots, x_n\} \in{\mathbb X}_{m,n}, n\geq{m}$. Hence,  ${\cal B}({\mathbb X}_{1,n})\triangleq\otimes_{t\in\mathbb{Z}_+^n}{\cal B}({\mathbb X}_t)$ denotes  the $\sigma-$algebra on ${\mathbb X}_{1,n}$ generated by cylinder sets $\{\{x_1,\ldots, x_n, \ldots,\}\in{\mathbb X}_{1,\infty}\big|x_j\in{A}_j, A_j\in{\cal B}({\mathbb X}_j),~j\in\mathbb{Z}_+^n\}$. We use the convention $X_{k,n}=X_{\max\{k,1\}, n}$ and $X_{k,n}=\{\emptyset\}, \forall k>n, (k,n)\in {\mathbb Z}\times {\mathbb Z}$. 
 \\
Given a measurable  space $\big(\Omega, {\cal F}\big)$, we denote the set of probability measures (PMs) ${\mathbb P}$ on $\Omega$  by ${\cal M}(\Omega)$.  Given a sequence of  RVs indexed by subscript $t$,   $X_t:  (\Omega, {\cal F}) \rar  ({\mathbb X}_t, {\cal B}({\mathbb X}_t)), \forall t \in {\mathbb Z}_+^n$,  we denote  by
${\mathbb P}\big\{X_1 \in dx_1, \ldots, X_n \in dx_n\big\}={\bf P}_{X_{1,n}}(dx_{1,n})\equiv {\bf P}(dx_{1,n})$ the PM induced by $X_{1,n}$ on $({\mathbb X}_{1,n}, {\cal B}({\mathbb X}_{1,n}))$ (i.e., probability distribution if ${\mathbb X}_t= {\mathbb R}^k$). Given another  RV, $Y: (\Omega, {\cal F}) \rar  ({\mathbb Y}, {\cal B}({\mathbb Y}))$ we define    the conditional PM  of the  RV $Y$ conditioned on   $X$ by $
{\bf P}_{Y|X}(dy| X)\tri {\mathbb P}\big\{Y \in dy\big|X\big\}
  \equiv {\bf P}(dy|X)$.\\
We identify conditional PMs  by stochastic kernels, $K:  {\cal B}({\mathbb Y }) \times {\mathbb X} \rightarrow [0,1]$ satisfying the  two properties     \cite{hernandezlerma-lasserre1996}: \\
1) $\forall x \in {\mathbb X}$, the set function $K(\cdot|x)$ is a probability measure
on $({\mathbb Y}, {\cal B}({\mathbb Y}))$, i.e.,  $K(\cdot|x)\in {\cal M}({\mathbb Y})$. \\
2) $\forall A \in {\cal B}({\mathbb Y})$, the function $K(A|\cdot)$ is
${\cal B}({\mathbb X})$-measurable. \\
We denote the set of such  maps $K(\cdot|\cdot)$  by  ${\cal K}({\mathbb Y}|{\mathbb X})$.

 \subsection{The Model of  Decentralized Stochastic Network}
 \label{sect:int-model}
We consider  a general Markovian stochastic network on the  finite horizon 
 ${T}_+^{n}\tri \{1, \ldots, n\}$, 
with observation posts, ${\mathbb Z}_+^K=\{1,2, \ldots, K\}$ collecting information, each associated with a   control station applying control actions,   described  by  the following elements. 

(a) The underlying  probability space $(\Omega, {\cal F}, {\mathbb P})$. 

(b) The    state process, $X_{1,n} \tri \big\{X_1,\ldots, X_n\big\}$, 
\begin{align*}
 X_t: (\Omega, {\cal F}) \rar ({\mathbb X}_t, {\cal B}({\mathbb X}_t)),\hso \forall t \in {T}_+^{n}. 
\end{align*}

(c)  The $K$ observation processes,  $Y_{1,n}^{(K)} \tri \big\{Y_{1,n}^1,  \ldots, Y_{1,n}^K\big\}$,  
\begin{align*}
Y_t^k: (\Omega, {\cal F}) \rar ({\mathbb Y}_t^k,{\cal B}(    {\mathbb Y}_t^k   )), \hso \forall t \in {T}_+^{n},  \hso \forall k \in {\mathbb Z}_+^K.
\end{align*} 

(d) The $K$ control processes,  $U_{1,n}^{(K)} \tri \big\{U_{1,n}^1,\ldots, U_{1,n}^K\big\}$,  
\begin{align*}
U_t^k: (\Omega, {\cal F}) \rar ({\mathbb U}_t^k,{\cal B}(  {\mathbb U}_t^k   )), \hso  \forall t \in {T}_+^{n} , \hso \forall k \in {\mathbb Z}_+^K
\end{align*}
where  $\big\{ {\mathbb U}_t^k| t \in {T}_+^{n}\big\}$ are referred to as the {\it  action spaces }of the controls $U_{1,n}^k=u_{1,n}^k\in \prod_{k=1}^K {\mathbb U}_t^k$,  $ \forall k \in {\mathbb Z}_+^K$. 

(e) The $T-$step delayed sharing information  patterns. 
For each $(t, k)\in T_+^{n} \times {\mathbb Z}_+^K$, 
let  $\sr{\circ}{I}_t^k\tri \{Y_{1,t}^k, U_{1,t-1}^k\}$ be the data set  of the $k$ control station   at time $t$,  and   consider its  decomposition   into  a private component $\Lambda_{t}^k$ and a $T-$step delayed shared component $\Delta_t^k$,  as follows.   
\begin{align}
&\sr{\circ}{I}_t^k\tri  \big\{Y_{1,t}^k, U_{1,t-1}^k\big\}= \Delta_t^k\cup  \Lambda_{t}^k, \; \forall (t, k)\in T_+^{n} \times {\mathbb Z}_+^K,  \label{inf_1}  \\
& \Delta_t^k \tri    \big\{Y_{1, t-T}^k, U_{1,t-T}^k \big\}, \hso T \in \{1,\ldots, n\}, 
\\
& \Lambda_t^k \tri     \big\{Y_{t-T+1, t}^k, U_{t-T+1, t-1}^k \big\}, \\
 &\Delta_t^k, : (\Omega, {\cal  F}) \rar ({\mathbb D}_t^k, {\cal B}({\mathbb D}_t^k))  \; \mbox{shared comp. to all controls}\nonumber \\
& 
 \Lambda_t^k: (\Omega, {\cal F}) \rar ({\mathbb L}_t^k, {\cal B}({\mathbb L}_t^k)) \;  \mbox{private comp. of control $U_t^k$}\nonumber
\end{align}
where  $T   $ is an arbitrary delay (positive integer). \\
For each $(t, k)\in T_+^{n} \times {\mathbb Z}_+^K$, the information pattern assigned to  control station  $k$ at time $t$ 
 is $I_t^k$   defined by 
  \begin{align}
  & I_t^k: (\Omega, {\cal F}) \rar ({\mathbb I}_t^k, {\cal B}({\mathbb I}_t^k)), \hso \forall (t, k)\in T_+^{n} \times {\mathbb Z}_+^K, \label{ISG-1}\\
&I_t^k \tri \sr{\circ}{I}_t^k \cup \Delta_t^{-k}= \Delta_t^k\cup \Lambda_t^k\cup \Delta_t^{-k} = \Lambda_t^k \cup \Delta_t^{(K)},      \\
&\Delta_t^{-k}\tri \big\{\Delta_{t}^1, \ldots, \Delta_t^{k-1},\Delta_t^{k+1}, \ldots, \Delta_t^{K} \big\}, \\
&\Delta_t^{(K)}\tri \big\{\Delta_{t}^1, \ldots, \Delta_t^K\big\}= \big\{ \Delta_t^k, \Delta_t^{-k}\big\}, \\
&\Lambda_t^{(K)}\tri \big\{\Lambda_{t}^1, \ldots, \Lambda_t^K\big\}= \big\{ \Lambda_t^k, \Lambda_t^{-k}\big\}.\label{ISG-2}
\end{align}
Hence,  $\Delta_t^{(K)}\in  I_t^k, \forall k$  is the shared information component available to all control stations,  and  $\Lambda_t^k \in I_t^k$  
 is the private information component available  only to  control station $k$..   
\ \

(f) The strategies of the Control.  
For each $(t, k)\in T_+^{n} \times {\mathbb Z}_+^K$,  the control $U_t^k$ applied by the $k$ control station at time $t$
is generated by a  measurable function $\gamma_t^k(\cdot)$  called the strategy, 
\begin{align}
U_t^k=\gamma_t^k(I_t^k)=\gamma_t^k(\Delta_t^{(K)}, \Lambda_t^k),  \;     \forall t \in {T}_+^{n},\;   k \in {\mathbb Z}_+^K. \label{type_a}
\end{align}
 For each $k \in {\mathbb Z}_+^K$,  we denote by  ${\cal  U}_{1,n}^{k}$ the set of  such  admissible  strategies of   control  $U_{1,n}^k$ over  $T_+^{n}$.
 
%
 

 \begin{remark} (On Delayed Sharing Patterns)\\
\label{rem:dsp}
\indent (1)  For any  $ t\in T_+^n$,
if 
 $t < T+1$ then $\Delta_t^{(K)}=\{\emptyset\}$,   $\Lambda_t^{k}  = \{Y_{1,t}^k,  U_{1, t-1}^k\}$, 
$U_t^k=\gamma_t^k(Y_{1,t}^k,  U_{1, t-1}^k) $, $\forall  (t,k)$.

{\it  (2) Information Patterns Without Sharing.}
  If $T=n$ (by (1)) then $\Delta_t^k=\{\emptyset\}$, $\Lambda_t^k= \{Y_{1,t}^k,  U_{1, t-1}^k\}$, and $I_t^k=\sr{\circ}{I}_t^k \tri   \Lambda_t^{k}$, $\forall  t\in T_+^n, \forall  k \in {\mathbb Z}_+^K$,  i.e., no information is shared. 

%

{\it   (3)  1-Step Delayed Sharing Patterns $T=1$.} For $T=1$ then $\Delta_t^{(K)}=\{Y_{1,t-1}^{(K)},  U_{1, t-1}^{(K)}\}$, $\Lambda_t^k= \{Y_{t}^k\}$, $\forall  t\in T_+^n, \forall  k \in {\mathbb Z}_+^K$.
 

\end{remark}

\ \

 Next, we introduce the various Markovian conditional PMs. 

{\it 1)   The Conditional  PM} of $X_{t+1}$ conditioned on $(X_{1,t}, Y_{1,t}^{(K)}, U_{1,t}^{(K)})$ is
\begin{align}
& {\mathbb P} \Big\{X_{t+1} \in dx_{t+1} \big| X_{1,t},Y_{1,t}^{(K)},   U_{1,t}^{(K)}\Big\}, \;  \;   \forall t \in  T_+^{n-1}\nonumber \\
&=
 {\bf P}_{X_{t+1}|X_{t}, U_{t}^{(K)}}=S_{t+1}(dx_{t+1}|X_t, U_{t}^{(K)}).  \label{i-d1}
 \end{align}
That is, the conditional PM is identified by the stochastic kernels   $S_{t+1}(\cdot|\cdot)\in  {\cal K}({\mathbb X}_{t+1}|{\mathbb X}_{t} \times {\mathbb U}_{t}^{(K)}), \forall t$.     

{\it  2) The Conditional PM} of $Y_{t}^k$ conditioned on $(X_{1,t},Y_{1,t-1}^{(K)}, Y_t^{-k}, U_{1,t-1}^{(K)})$ is  
\begin{align}
 &{\mathbb P}\Big\{Y_{t}^k \in dy_t^k \big| X_{1,t},Y_{1,t-1}^{(K)},Y_t^{-k}, U_{1,t-1}^{(K)}\Big\}, \;   \forall  t\in T_+^n, k \in {\mathbb Z}_+^K \nonumber \\
 &=
 {\bf P}_{Y_{t}^k|X_{t}, U_{t-1}^{(K)}} =Q_{t}^k( dy_t^k|X_t,  U_{t-1}^{(K)})   . \label{i-d2}
 \end{align} 
Hence,    $Q_{t}^k(\cdot|\cdot)\in  {\cal K}({\mathbb Y}_{t}^k|{\mathbb X}_{t}\times {\mathbb U}_{t-1}^{(K)})$. By Bayes' Theorem, 
\begin{align}
 &{\mathbb P} \Big\{Y_{t}^{(K)} \in dy_t^{(K)} \big| X_{1,t},Y_{1,t-1}^{(K)},  U_{1,t-1}^{(K)}\Big\}, \; \forall  t\in T_+^n  \label{i-d2-j-2}  \\
 &= Q_{t}^{(K)}(dy_t^{(K)} |X_t,  U_{t-1}^{(K)}) = \prod_{k=1}^K Q_{t}^k(dy_t^k |X_t,U_{t-1}^{(K)} ). \nonumber
 \end{align}
\indent {\it 3)     The Conditional} PM of $U_t^{(K)}$ conditioned on $(X_{1,t},Y_{1,t}^{(K)}, U_{1,t-1}^{(K)})$, for   $\gamma^{(K)}(\cdot) \in {\cal U} _{1,n}^{(K)}$ is 
\begin{align}
 &{\mathbb P} \Big\{U_{t}^{(K)} \in du_t^{(K)} \big| X_{1,t}, Y_{1,t}^{(K)}, U_{1,t-1}^{(K)}  \Big\}, \; \forall  t\in T_+^{n} \nonumber \\
  &=
 {\bf P}_{U_{t}^{(K)}|Y_{1,t}^{(K)},  U_{1,t-1}^{(K)}}= P_{t}^{(K)}(du_t^{(K)} |I_t^{(K)})     \label{i-d2-j-22} \\
& = \prod_{k=1}^K P_{t}^k(du_t^k |I_t^k)=\prod_{k=1}^K \mu_{\gamma_t^k(I_t^k)}(du_t^k )
 \label{i-d2-j-22-a}
 \end{align}
where  $P_{t}^k(\cdot|\cdot)\in  {\cal K}({\mathbb U}_{t}^k|{\mathbb I}_{t}^k)$,  
$P_{t}^k(du_t^k|I_t^k)=\mu_{\gamma_t^{k}(I_t^k)}(du_t^{k})$ is the  Dirac measure at  $\gamma_t^{k}(I_t^k) $.
From  (\ref{i-d2-j-22-a})  we deduce the {\it conditional independence (CI)}  or {\it Markov Chain (MC)},
\begin{align}
 &\mbox{\bf CI:} \hso {\bf P}_{U_{t}^{k}  | X_{1,t},I_{t}^{(K)},  U_t^{-k}}= P_{t}^{k}(du_t^k |I_t^k)\label{CI}  \\
 &
\hst \hst = \mu_{\gamma_t^k(I_t^k)}(du_t^k ),  \hso \forall  t\in T_+^{n}, \: \forall k \in {\mathbb Z}_+^K, \nonumber  \\
 &{\bf MC:} \hso (X_{1,t}, U_t^{-k}, I_t^{-k}) \leftrightarrow I_t^k \leftrightarrow U_t^k,\hso   \forall  t, \: \forall k. \label{MC}
 \end{align}

%


For notational simplicity we use the definitions ${\cal  U}_{1,n}^{k} \tri \prod_{t=1 }^n {\cal  U}_{t}^k, {\cal  U}_{1,n}^{-k} \tri \prod_{j=1,j\neq k }^K {\cal  U}_{1,n}^j$ and 
\begin{align} 
& \gamma_{t}^{-k}(I_t^{-k})
 \tri   \big\{\gamma_{t}^{1}(I_t^{1}),\ldots,  \gamma_{t}^{k-1}(I_t^{k-1}),   \gamma_{t}^{k+1}(I_t^{k+1}), \ldots, \nonumber \\
 &   \gamma_{t}^{K}(I_t^{K})\big\}, \hso  I_t^k=\Delta_t^{(K)}\cup \Lambda_t^k,    \hso   \gamma_{1,n}^{-k}(\cdot)   \in {\cal  U}_{1,n}^{-k},     \label{not_1a} \\
 &\equiv \gamma_t^{-k}(\Delta_t^{(K)}, \Lambda_t^{-k})\equiv \gamma_t^{-k},  \;  \forall t \in { T}_+^{n},   \;    \forall k \in {\mathbb Z}_+^K.\label{not_3}
 \end{align} 
 
\indent {\it 4) The Joint Probability Measure} is defined on  the space  of admissible histories,   $
 {\mathbb   G}_{1,n}\tri  (\mathbb{Y}_{1,n}^{(K)} \times {\mathbb X}_{1,n} \times \mathbb{U}_{1,n}^{(K)}), \forall n \in {\mathbb Z}_+^{n}
$, which is equipped with  the natural $\sigma$-algebra 
  ${\cal B}({\mathbb  G}_{1,n})$. 
For $\gamma^{(K)}(\cdot) \in {\cal U}_{1,n}^{(K)}$,   we define the joint PM,  ${\bf P}_{Y_{1,n}^{(K)},  X_{1,n}, U_{1,n}^{(K)}}^{\gamma^{(K)}}$ 
on  the  canonical space $\big({\mathbb   G}_{1,n}, {\cal  B}({\mathbb  G}_{1,n})\big)$, and  we  construct a probability space $\big(\Omega, {\cal F}, {\mathbb P}^{\gamma^{(K)}}\big)$ carrying the RVs $(Y_{1,n}^{(K)}, X_{1,n}, U_{1,n}^{(K)})$, as follows.  
\begin{align}
&{\bf P}_{Y_{1,n}^{(K)},  X_{1,n}, U_{1,n}^{(K)}}^{\gamma^{(K)}}\\
&=\prod_{k=1}^K P_{n}^{k}(du_{n}^{k} |I_{n}^{k}) Q_n^{(K)}(dy_{n}^{(K)} \big|X_n, U_{n-1}^{(K)}) \nonumber \\
&. 
  S_n(dx_n   \big| X_{n-1}, U_{n-1}^{(K)})  
    \ldots  \prod_{k=1}^K  P_{2}^{k}(du_{2}^{k} |I_{2}^{k})     \nonumber \\
& .  Q_2^{(K)}(dy_{2}^{(K)}\big|X_2, U_1^{(K)} )S_2(dx_2   \big| X_{1}, U_{1}^{(K)})  \nonumber \\
  &\prod_{k=1}^K   P_{1}^{k}(du_1^{k} |I_1^{k})   Q_1^{(K)}(dy_{1}^{(K)}\big|X_1 )  S_1(dx_1) \label{jd-o}
\end{align}
such that (\ref{i-d1})-(\ref{i-d2-j-22}) and CI,  (\ref{CI}), holds.  


 {\it 5) The Average  Payoff or Cost Function.}
  Given  a   $K$-tuple  $\gamma^{(K)}\tri \{\gamma^1, \ldots, \gamma^K\} \in {\cal  U}_{1,n}^{(K)} $, the average   payoff  is  
\begin{align}
&J_{n} (\gamma^{(K)}) \tri {\mathbb E}^{{\gamma^{(K)}}} \Big\{ \sum_{t=1}^{n} \ell(t, X_t, \gamma_t^{(K)}) \Big\}   \label{i8-cost_a}
\end{align}
where  $\ell(t, \cdot), \forall t$  is lower semicontinuous, bounded from below.

\indent To  deal with the asymmetry of information of the  strategies, we use   the definition   of  {\it person-by-person (PbP) optimality }
 from {\it static team theory},  \cite{marschak1955,radner1962,marschak-radner1972,krainak-speyer-marcus1982a,krainak-speyer-marcus1982b}.

\ \

 \begin{definition}(Decentralized Seq. Team Equilibrium) \\
 \label{def-pbp}
 The  $K-$tuple of strategies   $\gamma^{(K),o}\tri   \{\gamma^{1,o},  \ldots, \gamma^{K,o}\} \in {\cal  U}_{1,n}^{(K)}$
  is called  {\it decentralized sequential team equilibrium} (i.e., PbP optimal) if it satisfies, $ \forall k \in {\mathbb Z}_+^K$, 
\begin{align}
& {J}_{n}(\gamma^{k,o}, \gamma^{-k,o}) \leq  {J}_{n}(\gamma^{k}, \gamma^{-k,o}), \hso 
 \forall \gamma^k \in {\cal  U}_{1,n}^{k} \label{pbp}\\
&J_{n}(\gamma^{k,o}, \gamma^{-k,o}) \tri  \inf_{  \gamma^{k}\in {\cal U}_{1,n}^k}   {\mathbb E}^{^{\gamma^k, \gamma^{-k,o}}} \Big\{ \sum_{t=1}^{n} \ell(t, X_t,\gamma_t^k,  \gamma_t^{-k,o})\Big\} \label{opt-p-k}
\end{align}
$
\big\{\gamma_t^{k}, \gamma_t^{-k,o}\big\} \equiv\big\{\gamma_t^{k}(I_t^{k}),  \gamma_t^{-k,o}(I_t^{-k}) \big\}$  (see (\ref{not_1a})),  where ${J}_{n}(\gamma^{k,o}, \gamma^{-k,o})$ is the optimal payoff of  strategy $\gamma^k \in {\cal  U}_{1,n}^{k}$, when all other strategies are fixed to    $\gamma^{-k,o}\in {\cal U}_{1,n}^{-k}$. 

  \end{definition}

%
%

\ \

Decentralized team equilibrium  (\ref{pbp}) 
 is   analogous to the notion of classical Nash-equilibrium,  with the important difference that there is  only one  payoff common to all strategies,   and the information structures of the strategies are different. 

\ \

\begin{example}(POMDP Induced by Recursions)\\
\label{ex:pomdps-e}
\indent (1) The Markovian  conditional PMs  include PMs  induced by state and  observations  generated   by the   recursions,
\begin{align}
&X_{t+1}=f(t, X_t, U_t^{(K)}, W_t), \hso \forall t \in {T}_+^{n-1}, \label{NDM-3}\\
& Y_{t}^k   = h^k(t,X_t, U_{t-1}^{(K)}, V_t^k),\hso \forall t \in {T}_+^{n}, \; \forall k \in {\mathbb Z}_+^K \label{NDM-4} 
\end{align}
 $W_{1,n-1}$,  $V_{1,n}^{(K)}$, 
 are   exogenous  noises,  $\{X_1, W_1, V_1^k$, $\ldots, W_{n-1}, V_{n-1}^k, V_n^k|k\in {\mathbb Z}_+^K \}$ are mutually independent, 
%
and  
$f(\cdot), h^k(\cdot), \forall k$ are bounded measurable functions. 

(2) Witsenhausen \cite{witsenhausen1971} and Varaiya and Walrand \cite{varaiya-walrand1978} considered PMs induced by  slightly different Markovian recursions,  
\begin{align}
&X_{t}=f^w(t, X_{t-1}, U_t^{(K)}, W_t), \hso  t=1, \ldots, n, \label{NDM-3-w}\\
& Y_{t}^k   = h^{k,w}(t,X_{t-1}, V_t^k), \; \forall k \in {\mathbb Z}_+^K \label{NDM-4-w} 
\end{align}
$\{X_0, W_1, V_1^k, \ldots, W_{n}, V_{n}^k|k\in {\mathbb Z}_+^K \}$  mutually independent, $f^w(\cdot), h^{k,w}(\cdot)$  bounded measurable, 
i.e., with stoch. kernels  $S_t^w(dx_t|x_{t-1}, U_t^{(K)}), Q_t^{k,w}(dy_t|x_{t-1})$.\\ 
%
Our results are easily modified to include  (\ref{NDM-3-w}), (\ref{NDM-4-w}).
\end{example}

\section{DP Equations and  Information States}
\label{discrete}
We develop the DP approach from  first principles  using  Definition~\ref{def-pbp}.
This  leads to various  (four) distinct DP equations.

\subsection{Equivalent Payoffs and Preliminaries}
We start with certain preliminaries.  
In  Lemma~\ref{lemma:payoff-pbp}, for each $k\in {\mathbb Z}_+^K$,   we  express   the  payoff  of   strategy $\gamma^k \in {\cal U}_{1,n}^k$, i.e.,  $J_{n} (\gamma^{k}, \gamma^{-k,o})$,   w.r.t.  the \^a posteriori  PM of  nonlinear filtering of estimating the  extended state process $\big\{S_t^k \tri (X_{t},   \Lambda_t^{-k})|t \in T_+^n \big\}$, which is unknown to $\gamma^k(\cdot)$   from $\big\{I_t^k=\{\Delta_t^{(K)}, \Lambda_t^k\}| t \in T_+^n\big\}$. This  choice of extended state  $\big\{S_t^k\big| t \in T_+^n \big\}$ is  natural because  $\gamma^k$    does not observe   $S^k$. 

\ \

\begin{lemma}(Payoff of Decentralized Seq. Team Equilibrium)\\
\label{lemma:payoff-pbp}
For each $k \in {\mathbb Z}_+^K$, define the  \^a posteriori  PM conditional on  $I_t^k=\{\Delta_t^{(K)}, \Lambda_t^k\}$ by\footnote{Notation $\Xi_t^k[I_t^k](dx_t, d\lambda_{t}^{-k})$ is often used  as a reminder that this is a measurable function of the data $I_t^k$.} 
\begin{align}
& \Xi_t^k[I_t^k](dx_t, d\lambda_{t}^{-k})\tri {\mathbb P}^{{\gamma^k, \gamma^{-k}}} \Big\{X_{t} \in dx_t, \Lambda_t^{-k} \in d\lambda_t^{-k} \big| I_t^{k}\Big\}\nonumber \\
 &= {\bf P}_t^{{\gamma^k, \gamma^{-k}}} (dx_t, d\lambda_{t}^{-k}\big|I_t^k), \;  \forall t \in T_+^n, \; \forall k \in {\mathbb Z}_+^K. \label{pPM}
\end{align} 
The    payoff $J_{n} (\gamma^{k}, \gamma^{-k,o})$ of  strategy $\gamma^k \in {\cal U}_{1,n}^k$  for fixed    strategy  $\gamma^{-k}=\gamma^{-k,o}\in {\cal U}_{1,n}^{-k}$ is expressed as, 
\begin{align}
&J_{n} (\gamma^{k}, \gamma^{-k,o}) ={\mathbb E}^{^{\gamma^k, \gamma^{-k,o}}} \Big\{ \sum_{t=1}^{n} \ell(t, X_t,\gamma_t^k,  \gamma_t^{-k,o})  \Big\}\nonumber  \\
&=   {\mathbb E}^{^{\gamma^k, \gamma^{-k,o}}} \Big\{ \sum_{t=1}^{n}  \int_{{\mathbb X}_t \times {\mathbb L}_{t}^{-k}}   \ell(t, x_t,\gamma_t^k(I_t^{k}),  \gamma_t^{-k,o}(\Delta_t^{(K)}, \lambda_t^{-k}))\nonumber \\
&\hst .{\bf P}_t^{{\gamma^k, \gamma^{-k,o}}}(dx_t, d\lambda_{t}^{-k}\big|I_t^k) \Big\}, \hso \forall k \in {\mathbb Z}_+^K. \label{tpayf-pbp}
\end{align}
\end{lemma}
\begin{proof}  (\ref{tpayf-pbp}) follows   by 
re-conditioning on $I_t^k$.
\end{proof}

\ \


\noi Note that $J_{n} (\gamma^{k}, \gamma^{-k,o})$ in (\ref{tpayf-pbp})  is expressed in terms  of

\bi
\item  the {\it  private \^a posteriory PM of  the $k$ strategy}    $\gamma^k \in  {\cal U}_{1,n}^{k}$, i.e., $\big\{ \Xi_t^k[I_t^k]  \big|t \in T_+^n   \big\}$, and 

\item {\it the common/shared information component } $\big\{\Delta_t^{(K)}  \big|t \in T_+^n   \big\}$, available  to all strategies $\{\gamma^{k}, \gamma^{-k,o}\}\in  {\cal U}_{1,n}^{(K)}$. 
\ei 
For 
 $I_t^k= i_t^k$, then  $\Xi_t^k[I_t^k](\cdot)=\xi_t^k[i_t^k](\cdot)$,  hence  $\xi_t^k[\cdot](\cdot) \in {\cal K}( {\mathbb X}_t \times {\mathbb L}_{t}^{-k} \big| {\mathbb I}_t^k),$ i.e., a   stochastic kernel, and $i_t^k \longmapsto \xi_t^k[i_t^k](\cdot) \in  {\cal M}( {\mathbb X}_t \times {\mathbb L}_{t}^{-k})$ is a PM,  $\forall (t,k) $.

Lemma~\ref{lemma:nested} will  be  essential to our DP approach.

\ \
 
\begin{lemma}(Nested  Properties of Information Patterns)\\
\label{lemma:nested}
(1) The  nested properties  $\Delta_t^{(K)} \subseteq \Delta_{t+1}^{(K)}$  and   $I_t^k\subseteq I_{t+1}^k$ hold.
\begin{align}
\Delta_{t+1}^{(K)}
=&\big\{\Delta_{t}^{(K)}, Y_{t-T+1}^{k}, U_{t-T+1}^{k}, Y_{t-T+1}^{-k}, U_{t-T+1}^{-k}\big\} \label{nested_1}    \\
 \Lambda_{t+1}^k =& \big\{Y_{t-T+2, t+1}^{k}, U_{t-T+2,t}^{k}\big\}, \\
I_{t+1}^k 
=&\big\{I_{t}^k,Y_{t+1}^k, U_t^k, Y_{t-T+1}^{-k}, U_{t-T+1}^{-k}  \big\} \label{eq-nested-a-1} \\
=&  \big\{\Delta_{t+1}^{(K)}, \Lambda_{t+1}^k\big\}, \hso  I_t^k=\big\{\Delta_{t}^{(K)}, \Lambda_{t}^k\big\}.  \label{eq-nested}
\end{align}
(2) For each $(t, k)\in T_+^n\times {\mathbb Z}_+^K$ and $T \in \{1,\ldots, n\}$  consider $U_{t-T+1}^{-k}$ generated by strategies  
\begin{align}
U_{t-T+1}^{-k}=&  \gamma_{t-T+1}^{-k}(\Delta_{t-T+1}^{(K)}, \Lambda_{t-T+1}^{-k}) \\
=&\Big\{ \gamma_{t-T+1}^{j}(\Delta_{t-T+1}^{(K)}, \Lambda_{t-T+1}^{j})\Big\}_{j=1, j\neq k}^K , \label{dp_1_na-str1-old} \\
\Delta_{t-T+1}^{(K)}=&\big\{Y_{1, t-2T+1}^{(K)}, U_{1, t-2T+1}^{(K)}\big\}, \\
\Lambda_{t-T+1}^{j}=& \big\{Y_{t-2T+2, t-T}^{j}, U_{t-2T+2,t-T}^{j}, Y_{t-T+1}^{j}\big\}.
\end{align}
Then  $\Delta_{t-T+1}^{(K)}\subseteq \Delta_{t}^{(K)}$  and 
 $\Lambda_{t-T+1}^{j} \subseteq \{ \Delta_{t}^{(K)},Y_{t-T+1}^{j}\}, \forall  j \in {\mathbb Z}_+^K,  j \neq k, \forall t \in T_+^n$. That is, except $Y_{t-T+1}^{-k}$, all  the 
arguments  $\{ \Delta_{t-T+1}^{(K)}, \Lambda_{t-T+1}^{-k}\}$   of  $\gamma_{t-T+1}^{-k}(\cdot)$ (i.e., (\ref{dp_1_na-str1-old}))     are 
specified by $\Delta_{t}^{(K)}=  \{Y_{1, t-T}^{(K)}, U_{1, t-T}^{(K)} \}$.\\
(3) The Markov property holds, $\forall (t, k)$:
\begin{align*}
{\bf P}_{t+1}^{{\gamma^k, \gamma^{-k}}} (di_{t+1}^k\big|I_{1,t}^k, U_{1,t}^k)={\bf P}_{t+1}^{{\gamma^k, \gamma^{-k}}} ( di_{t+1}^{k}\big|I_t^k, U_t^k).   
\end{align*}
\end{lemma}
\begin{proof} Follows  from  the definitions.
\end{proof}

\ \

Our DP approach utilizes 
the  following two basic results. 
\bi
\item Theorem~ \ref{thm:is-pbp}, Lemma~\ref{lem:markov}:    {\it the private  \^a posteriori PMs,   $\big\{\Xi_t^k[I_{t}^k]| t \in T_+^n\big\}$, satisfy Markov recursions, conditioned  on the common information component}  $\big\{\Delta_t^{(K)}| t \in T_+^n\big\}$.
\ei

\noi We make use of  Theorem~ \ref{thm:is-pbp},  Lemma~\ref{lem:markov}, and  the classes of strategies of Definition~\ref{def:is-ss}, to derive various   DPs equations, and to identify  sufficient statistics for the  strategies.
%


\ \

\begin{definition} (Information States-Classes of Strategies)
\\ 
\label{def:is-ss}
Define a   centralized or  common  \^a posteriori PM of a  RV   $Z_{t}:\Omega \rar {\mathbb Z}$ conditioned on the common/shared data
 $\Delta_t^{(K)}$ by
%
 \begin{align}
  \Phi_{t}[\Delta_t^{(K)}](dz_{t}) \tri  {\bf P}_{t}^{\gamma^{(K)}} (dz_{t}\big|\Delta_t^{(K)}), \hso  \forall t \in T_+^n.\label{is-c}
 \end{align}
\indent (a) For each $k \in {\mathbb Z}_+^K$, the private  \^a posteriori PM $\big\{\Xi_t^k[I_t^k]\big| t \in T_+^n\big\}$   is  {\it  conditionally Markov (information state)} w.r.t. $\big\{\Delta_t^{(K)} \big|   t \in T_+^n\big\}$ (resp. $\big\{\Phi_t[\Delta_t^{(K)}] \big|  t \in T_+^{n}\big\}$),  if  the next state    $\Xi_{t+1}^k$ is determined from $\big\{\Xi_t^k, Y_{t+1}^k, U_t^k, \gamma_t^{-k}(\Delta_t^{(K)},\cdot) \big\}$  and  $\Delta_t^{(K)}$  (resp. $\Phi_t[\Delta_t^{(K)}] $), $\forall  t \in T_+^{n-1}$.

(b) For each $k \in {\mathbb Z}_+^k$, a statistic $\big\{{\cal G}_t^k|t \in { T}_+^n\big\}, {\cal G}_t^k : \Omega \rar {\mathbb G}_t^k$ is called a {\it sufficient statistic}   for the $k$ strategy $\gamma^k \in {\cal U}_{1,n}^k$ if    $U_t^k=g_t^k({\cal G}_t^k)$ for some measurable function $g_t^k(\cdot)$, $\forall t \in T_+^{n}$.  

(c) For each $k \in {\mathbb Z}_+^k$,  strategies $\gamma^k \in {\cal U}_{1,n}^k$  are   called,\\
 i) {\it semi-separated denoted  by ${\cal U}_{1,n}^{k,s-sep}$, }     if  $U_t^k=g_t^k(\Xi_t^k,$ $\Delta_t^{(K)},$ $ \Lambda_{t}^k)$, where $\big\{\Xi_t^k\big| t \in { T}_+^n\big\}$ is  {\it conditionally Markov} w.r.t. $\big\{\Delta_t^{(K)}\big|  t \in T_+^n\big\}$  and  $\big\{{\cal G}_t^k=\{\Xi_t^k,$ $\Delta_t^{(K)}, \Lambda_t^k \}| t \in T_+^n\big\}$ is a sufficient statistic  for $\gamma^k \in {\cal U}_{1,n}^k$ ; \\
 ii)  {\it separated   denoted  by ${\cal U}_{1,n}^{k,sep} \subseteq  {\cal U}_{1,n}^{k,s-sep}$, }   if  $U_t^k=g_t^k(\Xi_t^k, \Phi_t,\Lambda_t^k )$, where $\big\{\Xi_t^k\big|t \in T_+^n$ is {\it conditionally Markov} w.r.t.  $\big\{\Phi_t\equiv \Phi_t[\Delta_t^{(K)}]\big|  t \in T_+^{n}\big\}$  and 
 $\big\{{\cal G}_t^k=\{\Xi_t^k,\Phi_t, \Lambda_t^k \}| t \in T_+^n\big\}$ is a sufficient statistic  for $\gamma^k \in {\cal U}_{1,n}^k$; \\ 
iii) {\it information state      denoted  by ${\cal U}_{1,n}^{k,is} \subseteq  {\cal U}_{1,n}^{k,sep}$ }   if   it
 is the  restriction of ii) to  $U_t^k=g_t^k(\Xi_t^k, \Phi_t)$,   i.e., $g^k(\cdot)$ depends on $\Lambda_t^k$ through $\Xi_t^k \equiv \Xi_t^k[I_t^k], \forall t$. 
\end{definition}


\subsection{Private and Centralized Information  States}
\label{sect:pr-isr}

In Theorem~\ref{thm:is-pbp},  we derive a  recursion for   $\big\{\Xi_t^k[I_{t}^k]\big| t \in T_+^n\big\}$. 

In  Theorem~\ref{thm:is-cen}, we derive     two  Markov  recursions for the {\it centralized \^a posteriori PMs},  $\big\{\Phi_t[\Delta_t^{(K)}]\big|t\in T_+^n\big\}$.
 
 In Lemma~\ref{lem:markov}, 
we show that 
   $\big\{\Xi_t^k[I_{t}^k]| t \in T_+^n\big\}$   is Markov conditional on $\big\{\Delta_t^{(K)}\big|t\in T_+^n\big\}$. 
 
%
%



%

\ \

\begin{assumptions}(Absolute Continuity\footnote{These  are  standard assumptions   in    POMDPs  \cite{kumar-varayia:B1986,bertsekas-shreve1978}.})
\label{ass-1}
\begin{align*}
& S_{t+1}(dx_{t+1}|x_t, u_{t}^{(K)}) 
=s_{t+1}(x_{t+1}|x_t,  u_{t}^{(K)})dx_{t+1}, \; \forall t,  \\
&Q_{t}^k(dy_t^k|x_t,   u_{t-1}^{(K)})=q_{t}^k(y_t^k|x_t,   u_{t-1}^{(K)})dy_{t}^k, \; \forall (t,k)
\end{align*} 
i.e., $s_{t+1}(\cdot|x_t,  u_{t}^{(K)}),  q_{t}^k(\cdot|x_t,   u_{1,t}^{(K)})$ are   probability density functions (PDFs) (probability mass functions  for finite spaces).
\end{assumptions}


\ \

\begin{theorem}(Recursion of Private   \^a Posteriori PMs)\\
\label{thm:is-pbp}
Suppose  
Assumptions~\ref{ass-1} hold and let 
$T \in \{2,3,\ldots, n\}$. \\For each $ k  \in {\mathbb Z}_+^K$,  any   
 $\{\gamma_{1,n}^k,  \gamma_{1,n}^{-k}\}\in {\cal U}_{1,n}^k\times  {\cal U}_{1,n}^{-k}$,  and realization $I_t^k=i_t^k=\{\delta_t^{(K)}, \lambda_t^k\}$, 
  the  private 
   \^a posteriori PM,  
   $\Xi_t^k[I_t^k]=\xi_t^k[i_t^k]\equiv {\bf P}_t^{{\gamma^k, \gamma^{-k}}} (dx_t, d\lambda_{t}^{-k}\big|I_t^k=i_t^k),  \forall t \in T_+^n $  satisfies   the recursion,
 \begin{align}    
&{\bf P}_{t+1}^{{\gamma^k, \gamma^{-k}}} (dx_{t+1}, d\lambda_{t+1}^{-k}\big|i_{t+1}^k), \; \forall t \in {\mathbb T}_+^{n-1}, \forall k \in {\mathbb Z}_+^K \label{apost_1}\  \\
&= {\bf T}_{t+1}^{k}[y_{t+1}^{k}, u_t^k, \gamma_t^{-k}(\delta_t^{(K)}, \cdot),{\bf P}_{t}^{{\gamma^k, \gamma^{-k}}} (\cdot\big|i_{t}^k)](dx_{t+1}, d\lambda_{t+1}^{-k}) , \nonumber \\
& {\bf P}_1^{{\gamma^k, \gamma^{-k}}}(dx_1, d\lambda_1^{-k}|i_1^k)= {\bf P}_{X_1, Y_1^{-k}|Y_1^k=y_1^k}  
\end{align}
where $u_t^k=\gamma_t^k(i_t^k)$ and  the operator ${\bf T}_{t+1}^{k}[\cdot](\cdot,\cdot)$  is defined  by  (non-zero if its denominator  is non-zero and zero otherwise),
\begin{align}
&{\bf T}_{t+1}^{k}[y_{t+1}^{k}, u_{t}^k,  \gamma_t^{-k}(\delta_t^{(K)}, \cdot),\xi_{t}^{k} (\cdot)](dx_{t+1}, d\lambda_{t+1}^{-k}) \tri \nonumber  \\
&\frac{\overline{\bf T}_{t+1}^k[y_{t+1}^{k}, u_{t}^k, \gamma_t^{-k}(\delta_t^{(K)}, \cdot),\xi_{t}^{k} (\cdot)](dx_{t+1}, d\lambda_{t+1}^{-k}) }{\int_{{\mathbb X}_{t+1}  \times {\mathbb L}_{ t+1}^{-k}  } \overline{\bf T}_{t+1}^k[y_{t+1}^{ k}, u_{t}^k, \gamma_t^{-k}(\delta_t^{(K)},\cdot),\xi_{t}^{k} (\cdot)](dx_{t+1},d\lambda_{t+1}^{-k} )   }, \nonumber  \\
&\overline{\bf T}_{t+1}^k[y_{t+1}^{k}, u_{t}^k, \gamma_t^{-k}(\delta_t^{(K)}, \cdot),\xi_{t}^{k} (\cdot)](dx_{t+1},d\lambda_{t+1}^{-k}) \nonumber\\\
&\tri       \int_{{\mathbb X}_t} Q_{t+1}^{k}(dy_{t+1}^{k}|x_{t+1}, u_t^k, \gamma_t^{-k}(\delta_t^{(K)}, \lambda_t^{-k}))  \nonumber \\
&.\prod_{j=1, j\neq k}^K Q_{t+1}^{j}(dy_{t+1}^{j}|x_{t+1}, u_t^k, \gamma_t^{-k}(\delta_t^{(K)},\lambda_t^{-k})) \nonumber\\ 
&. S_{t+1}(dx_{t+1}\big|x_{t},u_{t}^k, \gamma_t^{-k}(\delta_t^{(K)}, \lambda_t^{-k})) \nonumber \\
& . \prod_{j=1, j \neq k}^K P_t^j(du_t^{j}\big| \delta_t^{(K)} , \lambda_t^j) \; \xi_{t}^{k} [i_t^k](dx_{t}, d\lambda_{t}^{-k}) \label{apost_1a} 
\end{align}
$P_t^j(du_t^{j}\big| \delta_t^{(K)}, \lambda_t^j)  = { \mu}_{\gamma_t^{j}(\delta_t^{(K)}, \lambda_t^j)}(du_t^{j})$  (the  Dirac PM).\\
For  $T=1$, 
  the term $\prod_{j=1, j \neq k}^K P_t^j(du_t^{j}\big| \delta_t^{(K)}, \lambda_t^j)$ is removed  (i.e.,  $\lambda_{t+1}^{-k}=Y_{t+1}^{-k}$). 
\end{theorem} 
\begin{proof} See Appendix~\ref{app:thm:is-pbp}.
\end{proof}


\ \

By Theorem~\ref{thm:is-pbp},  
   $\xi_{t+1}^k[i_{t+1}^k]$, is updated
by     using $\xi_{t}^k[i_{t}^k]$,   observation $y_{t+1}^k$,  the value 
 of  the control action  $u_t^k=\gamma_t^k(i_t^k)$ at $i_t^k$ only, and the  strategies $\gamma_t^{-k}(\delta_t^{(K)}, \cdot)$ evaluated at $\delta_t^{(K)}$. This leads to the  property of  Lemma~\ref{lem:ind}, which is  a generalization of  a known  property of    centralized POMDP \cite{striebel1965,kumar-varayia:B1986}.


\ \

\begin{lemma}(Indep. of  Private   \^a Post. PM on its Strategy)\\
\label{lem:ind}
For each $\forall k$, 
 $\xi_{t}^k[i_{t}^k]\equiv  {\bf P}_{t}^{{\gamma^k, \gamma^{-k}}} (dx_{t}, d\lambda_{t}^{-k}\big|i_{t}^k),\forall  t$ depends on the actions $u_{1,n}^k\in {\mathbb U}_{1,n}^k$ and not   the strategies  $ \gamma^k(\cdot)\in {\cal  U}_{1,n}^k$,  
\begin{align}
{\bf P}_{t}^{{\gamma^k, \gamma^{-k}}} (dx_{t}, d\lambda_{t}^{-k}\big|i_{t}^k)={\bf P}_{t}^{{\gamma^{-k}}} (dx_{t}, d\lambda_{t}^{-k}\big|i_{t}^k), \; \forall (t,    \gamma^k).\nonumber 
\end{align}
%
%
\end{lemma}
\begin{proof} The proof is  similar to centralized POMDP \cite[Theorem~1]{striebel1965} or
\cite[Lemma~5.10, pp.81]{kumar-varayia:B1986} (using induction).  
%
%
\end{proof}



\ \


\begin{lemma}(Conditionally   Markov Property)\\
\label{lem:markov}
\indent (1)  The process  $\{ \Xi_t^{k}[I_t^k]| t \in T_+^n\}$ is   conditionally Markov     w.r.t. $\{ \Delta_t^{(K)} | t \in T_+^n\}$,  i.e.,  
 \begin{align}
&{\mathbb  P}^{\gamma^{k}, \gamma^{-k}}\Big\{\Xi_{t+1}^k \in  d\xi_{t+1}^k \Big| I_t^k=i_t^k, U_t^k=u_t^k\Big\} \label{ext-m}   \\
&={\bf P}_{t+1}^{\gamma^{-k}}(d\xi_{t+1}^k \Big| \Xi_t^k=\xi_t^k, \Delta_t^{(K)}=\delta_t^{(K)}, U_t^k=u_t^k),\forall t.\nonumber
\end{align}
Moreover,  
for any bounded continuous  $\psi:{\mathbb X}_{t+1} \times     {\mathbb L}_{t+1}^{-k} \rightarrow {\mathbb R}$,  
\begin{align}
&{\bf E}^{\gamma^k, \gamma^{-k}}\Big\{ \int_{{\mathbb X}_{t+1} \times     {\mathbb L}_{t+1}^{-k}} \psi(x_{t+1}, \lambda_{t+1}^{-k})  \nonumber \\
&. {\bf P}_{t+1}^{{\gamma^k, \gamma^{-k}}} (dx_{t+1}, d\lambda_{t+1}^{-k}
\big|I_{t+1}^k)  \Big|I_{t}^k, U_t^k \Big\} \label{Mar_REC} \\
&={\bf E}^{ \gamma^{-k}}\Big\{ \int_{{\mathbb X}_{t+1} \times     {\mathbb L}_{t+1}^{-k} } \psi(x_{t+1}, \lambda_{t+1}^{-k})    \label{Mar_REC_a_n}  \\
&.{\bf T}_{t+1}^{k}[Y_{t+1}^{k}, U_{t}^k, \gamma_t^{-k}(\Delta_t^{(K)}, \cdot),\Xi_{t}^{k} (\cdot\big)](dx_{t+1}, d\lambda_{t+1}^{-k}) \nonumber \\
&\hst   \Big|   \Xi_t^k[I_t^k],   \Delta_t^{(K)}, U_t^k \Big\} \nonumber 
\end{align}
\begin{align}
&={\bf E}^{ \gamma^{-k}}\Big\{ \int_{{\mathbb X}_{t+1} \times     {\mathbb L}_{t+1}^{-k}} \psi(x_{t+1}, \lambda_{t+1}^{-k})   \label{Mar_REC_n}\\
& . \Xi_{t+1}^k[I_{t+1}^k](dx_{t+1}, d\lambda_{t+1}^{-k})   \Big|   \Xi_t^k[I_t^k],\Delta_t^{(K)}, U_t^k  \Big\}, \;  \forall t, k  \nonumber 
\end{align}
where    the conditional expectation in (\ref{Mar_REC_a_n})  
 is w.r.t.  PM,
\begin{align}
&{\bf P}_{t+1}^{ \gamma^{-k}}(dy_{t+1}^k\Big|   I_{t}^k, U_t^k)  
={\bf P}_{t+1}^{\gamma^{-k}}(dy_{t+1}^k\Big|    \Xi_{t}^k[I_t^k], \Delta_t^{(K)}, U_t^k)\nonumber  \\
&
= \int_{{\mathbb X}_{t+1} \times {\mathbb X}_t \times {\mathbb L}_{t}^{-k}   } Q_{t+1}^k(dy_{t+1}^k\big|    x_{t+1}, U_t^k, \gamma_t^{-k}(\Delta_t^{(K)}, \lambda_t^{-k})) \nonumber \\
&{ S}_{t+1}(dx_{t+1}\big|   x_t, U_t^k,\gamma_t^{-k}(\Delta_t^{(K)}, \lambda_t^{-k}))
\Xi_{t}^k[I_t^k](dx_{t}, d\lambda_{t}^{-k}) . \label{Mar_REC_1_n}
\end{align}
\indent (2)  For separated strategies   $\gamma^{(K)}\in {\cal U}_{1,n}^{(K),sep}\subseteq {\cal U}_{1,n}^{(K)}$,  then    $\{ \Xi_t^{k}[I_t^k] | t \in T_+^n\}$ is conditionally  Markov w.r.t..   $\{ \Phi_t[\Delta_t^{(K)}]$ $ | t \in T_+^n\}$,
 i.e,  (\ref{ext-m}) holds with    $\Delta_t^{(K)}$ replaced by $ \Phi_t[\Delta_t^{(K)}]$.
\end{lemma}
\begin{proof} (1)  First, we  prove (\ref{Mar_REC_1_n}).
\begin{align}
&{\bf P}_{t}^{{\gamma^k, \gamma^{-k}}}(dy_{t+1}^k\big| I_t^k, U_t^k)\nonumber \\
&\sr{(a)}{=} \int_{{\mathbb X}_{t,t+1} \times {\mathbb L}_{t}^{-k}    } Q_{t+1}^k(dy_{t+1}^k\big|    x_{t+1}, U_t^{k}, \gamma_t^{-k}(\Delta_t^{(K)}, \lambda_t^{-k}))\nonumber \\
&. {\bf  P}_{t+1}^{\gamma^k, \gamma^{-k}}(dx_{t+1}\big|   I_t^k, U_t^k, x_t, \lambda_{t}^{-k}){\bf P}_{t}^{\gamma^k, \gamma^{-k}}(dx_{t}, d\lambda_{t}^{-k}\big|  I_t^k, U_t^k) \nonumber   \\
&\sr{(b)}{=} \int_{{\mathbb X}_{t, t+1}  \times {\mathbb L}_{t}^{-k}   } Q_{t+1}^k(dy_{t+1}^k\big|    x_{t+1}, U_t^k, \gamma_t^{-k}(\Delta_t^{(K)}, \lambda_t^{-k}))\nonumber \\
&. { S}_{t+1}(dx_{t+1}\big|   x_t, U_t^k,\gamma_t^{-k}(\Delta_t^{(K)}, \lambda_t^{-k})) \Xi_{t}^{k}[I_t^k](dx_{t}, d\lambda_{t}^{-k})
\nonumber 
\end{align}
where $(a), (b)$ are  due  to  $I_t^k=\{\Delta_t^{(K)}, \Lambda_t^k\}$,  $\lambda_t^{-k}$  specifies the arguments of  $\gamma_t^{-k}(\cdot)$, 
 the model   PMs,  and   conditional  independence,  and Lemma~\ref{lem:ind}.
    From $(b)$
   we obtain (\ref{Mar_REC})-(\ref{Mar_REC_1_n}), and    (\ref{ext-m}) is shown.
   (2) This is obvious.
\end{proof}

\ \

In Theorem~\ref{thm:is-cen},  we identify two candidates  of   centralized \^a posteriori PMs    for   $\big\{\Phi_t[\Delta_t^{(K)}]\big|t\in T_+^n\big\}$ that satisfy  recursions.


%

\ \

\begin{theorem}(Recursions  of Centralized \^a Posteriori PMs)\\
\label{thm:is-cen}
Suppose Assumnptions~\ref{ass-1} hold. 

(1) For any   
 $\gamma^{(K)}\in  {\cal U}_{1,n}^{(K)}$, and realization $\Delta_t^{(K)}=\delta_t^{(K)}$, 
the centralized \^a posteriori PM,   $ \Theta_{t}[\Delta_t^{(K)}] = \theta_{t}[\delta_t^{(K)}]  \tri {\bf P}_t^{\gamma^{(K)}} (dx_{t-T}\big|\Delta_t^{(K)}=\delta_t^{(K)})$, satisfies  the Markov recursion, 
 \begin{align}    
&{\bf P}_{t+1}^{\gamma^{(K)}} (dx_{t-T+1}\big|\delta_{t+1}^{(K)}) \hst   t=T+1, \ldots,  n \label{apost_c}\  \\
&= {\bf T}_{t+1}^{(K)}[y_{t-T+1}^{(K)},u_{t-T}^{(K)} ,{\bf P}_{t}^{\gamma^{(K)}} (\cdot\big|\delta_{t}^{(K)})](dx_{t-T+1}) , \nonumber \\
& {\bf P}_{t}^{\gamma^{(K)}}(dx_{t-T}|\delta_{t}^{(K)})\big|_{t=T+1}= {\bf P}_{X_1|Y_{1}^{(K)}=y_{1}^{(K)}}, \\  
&{\bf T}_{t+1}^{(K)}[y_{t-T+1}^{(K)},  u_{t-T}^{(K)},\theta_{t} (\cdot)](dx_{t-T+1})  \nonumber  \\
&\tri\frac{ \overline{\bf T}_{t+1}^{(K)}[y_{t-T+1}^{(K)}, u_{t-T}^{(K)}, \theta_{t} (\cdot)](dx_{t-T+1}) }  { \int_{{\mathbb X}_{t-T+1} } \overline{\bf T}_{t+1}^{(K)}[y_{t-T+1}^{ (K)}, u_{t-T}^{(K)}, \theta_{t} (\cdot)](dx_{t-T+1})  }, \nonumber  \\
&\overline{\bf T}_{t+1}^{(K)}[y_{t-T+1}^{(K)}, u_{t-T}^{(K)},\theta_{t} (\cdot)](dx_{t-T+1}) \nonumber\\
&\tri       \int_{{\mathbb X}_{t-T}} S_{t-T+1}(dx_{t-T+1}\big|x_{t-T},u_{t-T}^{(K)})\nonumber  \\
&. Q_{t-T+1}^{k}(dy_{t-T+1}^{k}|x_{t-T+1}, u_{t-T}^{(K)})  \nonumber \\
&. Q_{t-T+1}^{-k}(dy_{t-T+1}^{-k}|x_{t-T+1}, u_{t-T}^{(K)}) 
\theta_{t}[\delta^{(K)}] (dx_{t-T}) \label{ope_is} 
\end{align}
Moreover, ${\bf P}_{t}^{\gamma^{(K)}} (\cdot\big|\delta_{t}^{(K)})={\bf P}_{t} (\cdot \big|\delta_{t}^{(K)}),\forall \gamma^{(K)}(\cdot), \forall  t$, i.e., does not depend on the strategies $\gamma^{(K)}(\cdot)$. 

(2) For any   
 $\gamma^{(K)}\in  {\cal U}_{1,n}^{(K)}$ and   $\Delta_t^{(K)}=\delta_t^{(K)}$, 
the centralized \^a posteriori PM,   $ \Pi_{t}^{\gamma^{(K)}}[\Delta_t^{(K)}] = \pi_{t}^{\gamma^{(K)}}[\delta_t^{(K)}] \tri {\bf P}_t^{\gamma^{(K)}} (dx_{t}, d\lambda_{t}^{(K)}\big|\Delta_t^{(K)}=\delta_t^{(K)})$, satisfies the Markov recursion, 
 \begin{align}    
&{\bf P}_{t+1}^{\gamma^{(K)}} (dx_{t+1}, d\lambda_{t+1}^{(K)}\big|\delta_{t+1}^{(K)}) \hst  \forall t \in T_+^{n-1} \label{apost_c-1}\  \\
&= {\bf T}_{t+1}^{(K)}[y_{t+1}^{(K)}, \gamma_{t}^{(K)},{\bf P}_{t}^{\gamma^{(K)}} (\cdot\big|\delta_{t}^{(K)})](dx_{t+1},d\lambda_{t+1}^{(K)}) , \nonumber \\
& {\bf P}_{1}^{\gamma^{(K)}}(dx_{1},d\lambda_{1}^{(K)}|\delta_1^{(K)})={\bf P}_{X_1, Y_1^{(K)}| \Delta_1^{(K)}=\delta_1^{(K)}}, \\
&{\bf T}_{t+1}^{{(K)}}[y_{t+1}^{(K)}, \gamma_{t}^{(K)},\pi_{t}^{\gamma^{(K)}} (\cdot)](dx_{t+1},d\lambda_{t+1}^{(K)}) \nonumber  \\
& \tri\frac{ \overline{\bf T}_{t+1}^{{(K)}}[y_{t+1}^{(K)}, \gamma_{t}^{(K)}, \pi_{t}^{\gamma^{(K)}} (\cdot)](dx_{t+1},d\lambda_{t+1}^{(K)}) }  { \int_{{\mathbb S}_{t+1}^{(K)} } \overline{\bf T}_{t+1}^{{(K)}}[y_{t+1}^{ (K)}, \gamma_{t}^{(K)},\pi_{t}^{\gamma^{(K)}} (\cdot)](dx_{t+1},d\lambda_{t+1}^{(K)})  }, \nonumber \\
&\overline{\bf T}_{t+1}^{{(K)}}[y_{t+1}^{(K)}, \gamma_{t}^{(K)},\pi_{t}^{\gamma^{(K)}} (\cdot)](dx_{t+1}, d\lambda_{t+1}^{(K)})\nonumber \\
& \tri       \int_{{\mathbb X}_{t}} Q_{t+1}^{(K)}(dy_{t+1}^{(K)}|x_{t+1}, u_{t}^{(K)}) S_{t+1}(dx_{t+1}\big|x_{t},u_{t}^{(K)})   \nonumber \\
&.  \prod_{j=1}^KP_t^j(du_t^{j}\big| \delta_t^{(K)}, \lambda_t^j)   \pi_{t}^{\gamma^{(K)}}[\delta_t^{(K)}] (dx_{t},d\lambda_{t}^{(K)})\label{ope_is-1} 
\end{align}
where  ${\mathbb S}_{t+1}^{(K)}\tri {\mathbb X}_{t+1}\times {\mathbb L}_{t+1}^{(K)} $ and $\gamma_{t}^{(K)}=\{\gamma_{t}^{k}(\delta_t^{(K)}, \lambda_t^k)\}_{k=1}^K$, 
$ P_t^j(du_t^{j}\big| \delta_t^{(K)}, \lambda_t^j)  = { \mu}_{\gamma_t^{j}(\delta_t^{(K)}, \lambda_t^j)}(du_t^{j}), \forall j , \forall t$.
\end{theorem} 
\begin{proof} The derivations are similar to Theorem~\ref{thm:is-pbp}.
\end{proof}

Note that recursion  $\theta_{t+1}[\delta_{t+1}^{(K)}] $  depends  on the actions  $\{y_{t-T+1}^{(K)},  u_{t-T}^{(K)}\}$, while recursion $ \pi_{t+1}^{\gamma^{(K)}}[\delta_{t+1}^{(K)}]  $ depends on  the strategies  $\gamma_t^{(K)}(\cdot)\in {\cal U}_{t}^{(K)}$.

\subsection{Generalized   DP Equations, Private  Information States, Semi-Separated Strategies}
\label{sect:DP-VP}
In Theorem~\ref{thm:vp}, we give the first necessary and sufficient conditions for  decentralized team equilibrium,  using generalized  DP equations.
   We show the {\it structural property } that optimal strategies occur in the subset  ${\cal U}_{1,n}^{(K),s-sep}\subseteq 
{\cal U}_{1,n}^{(K)}$.


%

\ \

\begin{definition}(Decentralized Value Processes)\\
\label{def:ctg-nm}
Consider  the payoff $J_{n} (\gamma^{k}, \gamma^{-k,o})$ of  strategy $\gamma^k \in {\cal U}_{1,n}^k$  for fixed   optimal strategies  $\gamma^{-k}=\gamma^{-k,o}\in {\cal U}_{1,n}^{-k}$,  $\forall k \in {\mathbb Z}_+^K$.\\
The value process or  cost-to-go ${\cal V}_{t}^{\gamma^{k,o},\gamma^{-k,o}}(\cdot): {\mathbb I}_t^{k}\rar (-\infty,\infty)$, over  $\{t,t+1, \ldots, n\}$  of   $\gamma^k \in {\cal U}_{1,n}^k$,  when the optimal  $\gamma^k=\gamma^{k,o}\in {\cal U}_{1,t-1}^{k}$ is  used over $\{1,2, \ldots, t-1\}$, conditioned  on any  $I_t^k=\{\Delta_t^{(K)}, \Lambda_t^k\}=i_t^k=\{\delta_t^{(K)}, \lambda_t^k\}$
  is   
\begin{align}
 & {\cal V}_{t}^{\gamma^{k,o},\gamma^{-k,o}}(i_t^k) \tri   \inf_{  \gamma^{k}\in {\cal U}_{t,n}^k}  {J}_{t,n}^{\gamma^{k},\gamma^{-k,o}}(i_t^k),   \;    \forall t \in T_+^n, 
  \label{opt-ctg-g}\\
&{J}_{t,n}^{\gamma^{k},\gamma^{-k,o}}(i_t^k) 
 \tri   {\mathbb E}^{^{\gamma^k, \gamma^{-k,o}}} \Big\{ \nonumber \\
 &\hso \sum_{j=t}^{n} \ell(j, X_j,\gamma_j^k(I_j^k),  \gamma_j^{-k,o}(I_j^{-k})) \Big|i_t^k \Big\}, \;   \forall k \in {\mathbb Z}_+^K.    \label{opt-ctg-pbp-tc} 
\end{align}
\end{definition}

\ \

%

\begin{lemma}(Decentralized Value Processes)\\
\label{lemma:eqv}
 For each $k$, the value process of Definition~\ref{def:ctg-nm} satisfies,   
\begin{align}
& {\cal V}_{n}^{\gamma^{k,o},\gamma^{-k,o}}(i_n^k)   \label{v-pbp-a} \\
&= \inf_{  \gamma^{k}\in {\cal U}_{n}^{k}}    \Big\{  \int_{{\mathbb X}_n\times {\mathbb L}_n^{-k}}   \ell(n, x_n,\gamma_n^k(I_n^k), \gamma_n^{-k,o}(\delta_n^{(K)},\lambda
_n^{-k})) \nonumber \\
&\hst  .{\bf P}_n^{{\gamma^k, \gamma^{-k,o}}} (dx_n, d\lambda_n^{-k}\; \big|i_n^k)  \Big\} , \hst \hst    \forall k \in {\mathbb Z}_+^K \nonumber 
\end{align}
\begin{align}
& {\cal V}_{t}^{\gamma^{k,o},\gamma^{-k,o}}(i_t^k)  
=\inf_{  \gamma^{k}\in {\cal U}_{t,n}^{k}}  {\mathbb E}^{^{\gamma^k, \gamma^{-k,o}}}    \Big\{ \label{v-pbp} \\
&\hst
\sum_{j=t}^{n} \int_{{\mathbb X}_j\times {\mathbb L}_j^{-k}}   \ell(j, x_j,\gamma_j^k(I_j^k),  \gamma_j^{-k,o}(\Delta_j^{(K)},\lambda
_j^{-k}))\nonumber \\
&. {\bf P}_j^{{\gamma^k, \gamma^{-k,o}}} (dx_j, d\lambda_j^{-k}\; \big|I_j^k) \Big| \; i_t^k  \Big\}, \hso  \forall t \in T_+^{n-1}, \; \forall k \in {\mathbb Z}_+^K.  \nonumber 
\end{align}
\end{lemma}
\begin{proof} By 
 (\ref{eq-nested}),   $I_t^k \subseteq I_{t+1}^k$,  and re-conditioning. 
\end{proof}

\ \

Theorem~\ref{thm:vp} gives  DP equations 
 based on  
the information patterns or statistics $\{I_t^k\big|t \in T_+^n\}, \forall k$. 
These are analogous to the  centralized DP equations of POMDPs in   \cite[Section~5.4]{bertsekas2005}. 

\ \

\begin{theorem}(Decentralized generalized DP Equations-Necessary and Sufficient Conditions)\\
\label{thm:vp}
For each $ k \in {\mathbb Z}_+^K$,  consider  
   Definition~\ref{def:ctg-nm},  $\big\{{\cal V}_t^{\gamma^{k,o}, \gamma^{-k,o}}(\cdot)\big| t\in T_+^n  \big\}$  and $ \xi_t^k[i_t^k] \tri  {\bf P}_{t}^{{ \gamma^{-k,o}}} (dx_{t}, d\lambda_{t}^{-k}\big|I_{t}^k=i_t^k)$, $i_t^k=\{\delta_t^{(K)}, \lambda_t^k\},    \forall t \in T_+^n$. \\ 
(1) {\it Necessary Conditions for PbP Optimality.}  Suppose a decentralized team equilibrium  $\gamma^{k,o} \in {\cal U}_{t,n}^k, \forall k$ exists. \\
(1.1) 
${\cal V}_t^{\gamma^{k,o}, \gamma^{-k,o}}(i_t^k)=  {\cal V}_t^{ \gamma^{-k,o}}(i_t^k), \forall i_t^k,  \forall  \gamma^{k,o}  \in {\cal U}_{1,n}^k, \forall  t\in T_+^n $ (i.e., is independent  of $\gamma^{k,o}$) 
 and  satisfies   the   DP   equations: 
\begin{align}
&{\cal V}_n^{\gamma^{-k,o}}(i_n^k)=\inf_{u_n^k \in {\mathbb U}_n^k}{\mathbb E}^{ \gamma^{-k,o}}   \Big\{   \ell(n, X_n,u_n^k,  \gamma_n^{-k,o}(\delta_n^{(K)}, \Lambda_n^{-k}))   \nonumber \\
& \Big| i_n^k, u_n^k  \Big\},\hst  \forall k \in {\mathbb Z}_+^K \label{dp_1_nnn}\\
&= \inf_{u_n^k \in {\mathbb U}_n^k}\int_{{\mathbb X}_n  \times {\mathbb L}_{n}^{-k} }     \ell(n, x_n,u_n^k,  \gamma_n^{-k,o}(\delta_n^{(K)}, \lambda_n^{-k})) \nonumber \\
&\hst . \xi_n^k[i_n^k] (dx_{n}, d\lambda_n^{-k}),      \label{dp_1_nnn_1}  \\
&{\cal V}_t^{\gamma^{-k,o}}(i_t^k) 
 = \inf_{u_t^k \in {\mathbb U}_t^k}  {\mathbb E}^{ \gamma^{-k,o}}   \Big\{    \ell(t, X_t,u_t^k,  \gamma_t^{-k,o}(\delta_t^{(K)}, \Lambda_t^{-k}))  \nonumber\\
 &+{\cal V}_{t+1}^{\gamma^{-k,o}}(I_{t+1}^k ) \Big| i_t^k,  u_t^k \Big\}, \; \forall t \in T_+^{n-1},\; \forall  k \in {\mathbb Z}_+^K \label{dp_2_nm}   \\
 &=  \inf_{u_t^k \in {\mathbb U}_t^k}  \Big\{  \int_{{\mathbb X}_t \times {\mathbb L}_{t}^{-k}}    \ell(t, x_t,u_t^k,  \gamma_t^{-k,o}(\delta_t^{(K)}, \lambda_t^{-k})) \label{dp_1_na}\\
 &.  \xi_t^k[i_t^k]  (dx_{t}, d\lambda_{t}^{-k}) +  \int_{ {\mathbb Y}_{t+1}^k \times {\mathbb X}_{t, t+1}   \times {\mathbb L}_{t}^{-k}    } {\cal V}_{t+1}^{\gamma^{-k,o}}(i_{t}^k,  y_{t+1}^k, u_t^k,\nonumber\\
 & y_{t-T+1}^{-k},  u_{t-T+1}^{-k})   {\bf P}_{t+1}^{ \gamma^{-k,o}}(dy_{t+1}^k, d\lambda_{t}^{-k}, dx_{t}, dx_{t+1}\big|   i_t^k, u_t^k)\Big\}  \nonumber 
\end{align}
where in  (\ref{dp_1_na}) 
 the   conditional  PM and $u_{t-T+1}^{-k}$ are
\begin{align}
&{\bf P}_{t+1}^{ \gamma^{-k,o}}( dy_{t+1}^k, d\lambda_t^{-k},dx_{t}, dx_{t+1} \big|i_t^k, u_t^k) \label{Mar_REC_1-nn-11}  \\
&={\bf P}_{t+1}^{ \gamma^{-k,o}}( dy_{t+1}^k, d\lambda_t^{-k},dx_{t}, dx_{t+1} \big| \xi_t^k, \delta_t^{(K)},  u_t^k) \label{Mar_REC_1-nn-aaa}\\
&= Q_{t+1}^k(dy_{t+1}^k\big|    x_{t+1}, u_t^k, \gamma_t^{-k,o}(\delta_t^{(K)}, \lambda_t^{-k}))\label{Mar_REC_1-nn} \\
&.{ S}_{t+1}(dx_{t+1}\big|   x_t, u_t^k,\gamma_t^{-k,o}(\delta_t^{(K)}, \lambda_t^{-k})) \xi_t^k[i_t^k](dx_{t}, d\lambda_{t}^{-k}), \nonumber \\
 &u_{t-T+1}^{-k}=\gamma_{t-T+1}^{-k,o}(i_{t-T+1}^{-k})= \gamma_{t-T+1}^{-k,o}(\delta_{t-T+1}^{(K)}, \lambda_{t-T+1}^{-k})\nonumber \\
 &=\Big\{ \gamma_{t-T+1}^{j,o}(\delta_{t-T+1}^{(K)}, \lambda_{t-T+1}^{j})\Big\}_{j=1,j \neq k}^K, \label{dp_1_na-str1}
\end{align}
and the  arguments  $\{\delta_{t-T+1}^{(K)},\lambda_{t-T+1}^{-k}\}$    of  $\gamma_{t-T+1}^{-k,o}(\cdot)$, 
  except  $y_{t-T+1}^{-k}$,  are  specified by $\delta_{t}^{(K)}$ (see Lemma~\ref{lemma:nested}.(2)).  
\\
(1.2) If  the infimum in the DP equations of part (1.1) exists,  then $\gamma^{k,o} \in {\cal U}_{1,n}^{k,s-sep}\subseteq {\cal U}_{1,n}^{k}$,  $u_n^{k,o}=\gamma_n^{k,o}(\xi_n^k[i_n^k], \delta_n^{(K)})\in {\mathbb U}_n^k$,   $u_t^{k,o}=\gamma_t^{k,o}(\xi_t^k[i_t^k], \delta_t^{(K)}, \lambda_t^k)\in {\mathbb U}_t^k, \forall  t \in T_+^{n-1}, \forall k$. \\
(2) {\it Verification-Sufficient  Conditions for Decentralized Team Equilibrium.} \\
(2.1)  If the value process  ${\cal V}_t^{ \gamma^{-k,o}}(\cdot),  \forall t \in T_+^n, \forall k$  satisfy the DP equations  of part (1),  then  (almost surely), 
\begin{align}
 {\cal V}_t^{\gamma^{-k,o}}(I_t^k)    \leq J_{t,n}^{\gamma^{k}, \gamma^{-k,o}}(I_t^k) ,  \;\forall  \gamma^k \in {\cal U}_{1,n}^k, \; \forall (t,k)  \label{in_1_g} 
\end{align}
and  $\gamma^{k,o} \in {\cal U}_{1,n}^k, \; \forall k \in {\mathbb Z}_+^K$ is  decentralized team equilibrium.\\
(2.2) Suppose for each $k \in {\mathbb Z}_{+}^K$, 
$\gamma^{k,o}(\cdot)$ is a strategy $  u_t^{k,o}=\gamma_t^{k,o}(\xi_t^k[i_t^k], \delta_t^{(K)}, \lambda_t^k)\in {\mathbb U}_t^k$,  i.e., $\gamma^{k,o} \in   {\cal U}_{1,n}^{k}$, 
such that,  for all $\{\xi_t^k[i_t^k], \delta_t^{(K)}, \lambda_t^k\}$ achieves   the infimum in the DP eqns of part (1)  for $t=1, \ldots, n$. 
Then $\gamma^{k,o} \in  {\cal U}_{1,n}^{k}$ is  optimal and ${\cal V}_t^{\gamma^{-k,o}}(I_t^k) = J_{t,n}^{\gamma^{k,o}, \gamma^{-k,o}}(I_t^k), \forall (t, k)$  (almost surely).
\end{theorem}
\begin{proof} See Appendix~\ref{app;thm:vp}.
\end{proof}

\ \

 

\begin{remark}(Completing  Witsenhausen's \cite[Assertion 8]{witsenhausen1971})\\
Theorem~\ref{thm:vp} identifies   the  structural property  of optimal strategies,  $u_t^k=\gamma_t^k((\xi_t^k[i_t^k],\delta_t^{(K)}, \lambda_t^{k})$.   By using  $\theta_t[ \delta_t^{(K)}]={\bf P}_{t}^{\gamma^{(K)}}(dx_{t-T}\big| \delta_t^{(K)} )$ instead of $\delta_t^{(K)}$, we identify the stractural compression property $u_t^k=\gamma_t^{k,sep}(\xi_t^k[i_t^k],\theta_t[\delta_t^{(K)}], \lambda_t^{k}), \forall t,k$.
On the other hand,  Witsenhausen's \cite[Assertion 8, pp.1562]{witsenhausen1971} uses strategies  $u_t^k=\phi_t^{k,W}(\theta_t[\delta_t^{(K)}]  , \lambda_t^{k})$, i.e., $\xi_t^k[i_t^k]$ is not included,   $\forall t,k$, which turned out to be sub-optimal according to the counterexample  provided   by Varaiya and Walrand \cite{varaiya-walrand1978}. \\
 Essentially, Theorem~\ref{thm:vp}  corrects Witsenhausen's \cite[Assertion 8]{witsenhausen1971}.  This structural  property also holds if  we  consider the {\it single cost-to-go} based only on the common information component  \cite{nayyar-mahajan-teneketzis2011,nayyar-mahajan-teneketzis2013,nayyar-teneketzis2019} (see Section~\ref{sec:lr}).
\end{remark}

\begin{remark}(Simplified Strategies)\\
\label{rem:1-sim} 
In Theorem~\ref{thm:vp}, 
  $\gamma_{t-T+1,t-1}^{k,o}$ can be recursively substituted in $U_t^k=\gamma_t^{k,o}(I_t^k)=\gamma_t^{k,o}(\Delta_{t}^{(K)}, \Lambda_t^{-k}),\Lambda_t^{-k}=\{Y_{t-T+1,t}^{-k}$, $U_{t-T+1,t-1}^{-k}\}$ to express  (Striebel  \cite{striebel1965} and Witsenhausen\cite{witsenhausen1971}), 
\begin{align}
&U_1^k=\overline{\gamma}_1^{k,o}(\Delta_{1}^{(K)}, Y_{1}^k), \:
 U_2^k=\overline{\gamma}_2^{k,o}(\Delta_{2}^{(K)}, Y_{3-T,2}^k), \ldots, \nonumber \\
&  U_t^k=\overline{\gamma}_t^{k,o}(\Delta_{t}^{(K)},Y_{t-T+1,t}^{k} ), \; \forall t \in T_+^n,\ \;  \forall k \in {\mathbb Z}_+^K  
  \end{align} 
  for some measurable $\overline{\gamma}_{1,n}^{k,o}(\cdot) \in {\cal U}_{1,n}^k, \forall k$.   Clearly,  $\overline{\gamma}_t^k(\cdot)$ is not required to  recall  past controls $U_{t-T+1, t-1}^k$,     $\forall t \in T_+^{n}, \forall k$.  \\
Consequently, 
in 
 the DP equations of Theorem~\ref{thm:vp}  we can replace $(\gamma^k(\cdot), I_t^k)$ and the PMs   ${\bf P}_t^{{\gamma}^{-k,o}} (dx_t, di_{t}^{-k}\big|I_t^k)$ by  $(\overline{\gamma}_t^k(\cdot), \overline{I}_t^k\tri \{\Delta_t^{(K)}, Y_{t-T+1,t}^k\})$  and  the  PMs,  
  \begin{align}
 &\overline{\Xi}_t^k[\overline{I}_t^k](dx_{t}, dy_{t-T+1,t}^{-k}) \tri  {\bf P}_{t}^{{ \overline{\gamma}^{-k,o}}} (dx_{t}, dy_{t-T+1,t}^{-k}\big| \overline{I}_{t}^k) \nonumber  \\
& \forall t \in {\mathbb Z}_+^{n}, \; \forall k \in {\mathbb Z}_+^K.  
\end{align} 
\end{remark} 
%
%
%

\ \

\begin{remark} (On the DP Eqns of Theorem~\ref{thm:vp}).\\
\label{rem:papm-DP-nm}
\indent  For  $T=2$,   then   $\Delta_{t}^{(K)}=\{Y_{1, t-2}^{(K)}, U_{1, t-2}^{(K)}\}$,    $\Lambda_{t}^{k}=\{Y_{t-1,t}^{k},U_{t-1}^k\}$, $I_t^k=\{\Delta_{t}^{(K)}, \Lambda_{t}^{k}\}$.
   The conditional  expectation  in the RHS of the DP equation is
 \begin{align}
 & {\mathbb E}^{ \gamma^{-k,o}}   \Big\{   \ell(t, X_t,u_t^k, \gamma_t^{-k,o}(\delta_t^{(K)}, Y_{t-1,t}^{-k}, U_{t-1}^{-k}))\\
  &+  {\cal V}_{t+1}^{\gamma^{-k,o}}(i_{t}^k,Y_{t+1}^k, u_t^k, Y_{t-1}^{-k}, U_{t-1}^{-k} )\big|   i_t^k, u_t^k \Big\}, \nonumber \\
 & \gamma_t^{-k,o}(\delta_t^{(K)}, \Lambda_t^{-k} ) =\big\{ \gamma_t^{j,o}(\delta_t^{(K)}, Y_{t-1,t}^{j}, U_{t-1}^j )\big\}_{j=1, \j\neq k}^K, \nonumber   \\
&U_{t-1}^{-k}=\big\{ \gamma_{t-1}^{j,o}(\delta_{t-1}^{(K)},y_{t-2}^j,Y_{t-1}^{j},u_{t-2}^{j}) \big\}_{j=1, \j\neq k}^K
  \end{align}
and  $\xi_t^k[i_t^k] (dx_{t}, d\lambda_{t}^{-k}) =\xi_t^k[i_t^k] (dx_{t}, dy_{t-1, t}^{-k}, du_{t-1}^{-k})$. By Remark~\ref{rem:1-sim}, we can replace   $\{\gamma^{k,o}(\cdot), I_t^k)\}$ and $\Xi_t^k[I_t^k](dx_{t}, dy_{t-1, t}^{-k}, du_{t-1}^{-k})$ by   $\{\overline{\gamma}^{k,o}(\cdot), \overline{I}_t^k)\}$ and $\overline{\Xi}_t^k[\overline{I}_t^k](dx_{t}, dy_{t-1, t}^{-k}), \overline{I}_t^k=\{\Delta_t^{(K)}, Y_{t-1, t}^{k}\}$. 

\end{remark}

\ \

\begin{remark} (On Measurable Selectors-Theorem~\ref{thm:vp}).\\
\label{rem:papm-ex}
Theorem~\ref{thm:vp} pre-supposes the  assumptions:  For each $k$,  the terms inside the infimum of the DP equations, i.e., $\inf_{u_t^k \in {\mathbb U}_t^k} \big\{ \cdot \big\}$ is measurable,  the DP solution ${\cal V}_t^{\gamma^{-k,o}}(\cdot)$   is  measurable in $i_t^k$,  and  there exists measurable function $\gamma_t^{k,o}(\cdot) \in {\cal U}_{t}^k$ such that the function inside $\{ \cdot \}$ attains its minimum at   $u_t^{k,o} =\gamma_n^{k,o}(\xi_t^k,i_t^k ) \in {\mathbb U}_t^k, \forall (\xi_t^k,i_t^k)$, $\forall t$. \\
For   finite   spaces, 
  these assumptions are easy to verify \cite{bertsekas-shreve1978}.\\
 For abstract Borel spaces and Example~\ref{ex:pomdps-e} sufficient conditions are obtained from \cite{hernandezlerma-lasserre1996} \cite{bertsekas-shreve1978}, as follows.\\
1) ${\mathbb U}_t^k$ are compact Borel spaces, for $ \forall (t, k)$.\\
2) $\ell(t,x, \cdot)$ is lower semi-continuous bounded from below on ${\mathbb U}_t^{(K)}$ for every $(t,x) \in T_+^n\times {\mathbb X}_t$. \\
3) PMs  (\ref{Mar_REC_1-nn}) are such that  
\begin{align*}
& \phi^\prime(i_{t}^k, u_t^k) \tri   \int_{{\mathbb Y}_{t+1}^k \times {\mathbb X}_{t, t+1}   \times {\mathbb L}_{t}^{-k}    } \phi(y_{t+1}^k, \lambda_t^{-k}, x_{t}, x_{t+1}) \\
&. {\bf P}_{t+1}^{ \gamma^{-k,o}}(dy_{t+1}^k, d\lambda_{t}^{-k}, dx_{t}, dx_{t+1}\Big|   i_t^k, u_t^k)
\end{align*}
 are bounded continuous in  $(i_{t}^k, u_t^k)\in {\mathbb I}_t^k \times {\mathbb U}_t^k$ for every bounded continuous $\phi(\cdot)$ on  ${\mathbb Y}_{t+1}^k   \times {\mathbb L}_{t}^{-k}\times {\mathbb X}_{t, t+1}$, $\forall (t,k)$.   
\end{remark}

\ \

Theorem~\ref{thm:vp-mn}  simplifies  the  DP equations by    using  the  recursions  of the private \^a posteriori PMs  of Theorem~\ref{thm:is-pbp}.
These DP equations  are generalizations of the  centralized DP equations of POMDPs, which use the recursion of nonlinear filtering \^a posteriori PM, i.e.,  \cite[Section~5.4]{bertsekas2005}, \cite{kumar-varayia:B1986}, \cite{vanschuppen2021}.

\ \

\begin{theorem}(Decentralized DP Equations  with Private Information  State-Necessary Conditions)\\
\label{thm:vp-mn}
 Suppose Assumptions~\ref{ass-1} hold. 

(1)  For each $ k \in {\mathbb Z}_+^K$ suppose a PbP optimal $\gamma^{k,o} \in {\cal U}_{t,n}^k$ exists. 
The value processes of Definition~\ref{def:ctg-nm}  satisfy, $\forall i_t^k$, 
\begin{align}
 {\cal V}_{t}^{\gamma^{k,o},\gamma^{-k,o}}(i_t^k) =V_t^{\gamma^{-k,o}}(\xi_t^k,\delta_t^{(K)},\lambda_t^k ), \;  \forall (t,  k)   
\end{align}
where 
 $ { V}_t^{ \gamma^{-k,o}}(\cdot)$  
   satisfies  
 $\forall ( \xi_t^k, \delta_t^{(K)}, \lambda_t^k), \forall t \in T_+^n:$
\begin{align}
&{ V}_n^{\gamma^{-k,o}}(\xi_n^k,\delta_n^{(K)}, \lambda_n^k  )=\inf_{u_n^k \in {\mathbb U}_n^k}{\mathbb E}^{ \gamma^{-k,o}}   \Big\{ \label{dp_1_nnn-mn}   \\
& \ell(n, X_n,u_n^k,  \gamma_n^{-k,o}(\delta_n^{(K)}, \Lambda_n^{-k}))   \Big| \xi_n^k,\delta_n^{(K)}, \lambda_n^k  , u_n^k \Big\}, \forall k \in {\mathbb Z}_+^K \nonumber\\
&= \inf_{u_n^k \in {\mathbb U}_n^k}\int_{{\mathbb X}_n  \times {\mathbb L}_{n}^{-k} }     \ell(n, x_n,u_n^k,  \gamma_n^{-k,o}(\delta_n^{(K)}, \lambda_n^{-k})) \label{dp_1_nnn_1-mn}\\
&\hst . \xi_{n}^{k}[i_n^k] (dx_{n}, d\lambda_n^{-k}),     \\
&{ V}_t^{\gamma^{-k,o}}(\xi_t^k,\delta_t^{(K)},\lambda_t^k  )  = \inf_{u_t^k \in {\mathbb U}_t^k}  {\mathbb E}^{ \gamma^{-k,o}}   \Big\{ \label{dp_2_mn} \\
&   \ell(t, X_t,u_t^k,  \gamma_t^{-k,o}(\delta_t^{(K)}, \Lambda_t^{-k}))  +{ V}_{t+1}^{\gamma^{-k,o}}(\Xi_{t+1}^k, \Delta_{t+1}^{(K)}, \Lambda_{t+1}^k  )\nonumber\\
 & \Big| \xi_t^k,\delta_t^{(K)}, \lambda_t^k , U_t^k \Big\}, \hso  \forall t \in T_+^{n-1},  \forall k \in {\mathbb Z}_+^K,  \label{dp_2_mn}   \\
 &=  \inf_{u_t^k \in {\mathbb U}_t^k}  \Big\{  \int_{{\mathbb X}_t \times {\mathbb L}_{t}^{-k}}    \ell(t, x_t,u_t^k,  \gamma_t^{-k,o}(\delta_t^{(K)}, \lambda_t^{-k}))  \nonumber \\
 &.  \xi_{t}^k[i_t^k] (dx_{t}, d\lambda_{t}^{-k})  +  \int_{ {\mathbb Y}_{t+1}^k \times {\mathbb X}_{t, t+1}   \times {\mathbb L}_{t}^{-k}    }  {V}_{t+1}^{\gamma^{-k,o}}({\bf T}_{t+1}^{k}[y_{t+1}^{k},  \nonumber  \\
 &u_t^k, \gamma_t^{-k}(\delta_t^{(K)}, \cdot),    \xi_t^k(\cdot)], \delta_t^{(K)}, \lambda_t^k ,   y_{t+1}^k, u_t^k,    y_{t-T+1}^{-k}, u_{t-T+1}^{-k}) \nonumber  \\
 &. {\bf P}_{t+1}^{ \gamma^{-k,o}}(dy_{t+1}^k, d\lambda_{t}^{-k}, dx_{t}, dx_{t+1}\Big|   \xi_t^k, \delta_t^{(K)}, \lambda_t^k, u_t^k) \Big\}  \label{dp_1_na_m}
\end{align}
where   
 the PM of the last RHS term is
given by  (\ref{Mar_REC_1-nn-11})-(\ref{dp_1_na-str1}).
%

(2)  For each $k \in {\mathbb Z}_{+}^K$, if  the infimum in the DP equations of part (1) exists, then the optimal strategy occurs in the set of semi-separated strategies, $\gamma^{k,o} \in  {\cal U}_{1,n}^{k, s-sep}$, i..e., ${\cal G}_t^k\tri \{\Xi_t^k[I_t^k],\Delta_t^{(K)},\Lambda_t^k \}$ is a sufficient statistic for $U_t^k$, $\forall t$. In particular, 
 $u_n^{k,o}=\gamma_n^{k,o}(\xi_n^k[i_n^k],\delta_n^{(K)})\in {\mathbb U}_n^k$, 
$u_t^{k,o}=\gamma_t^{k,o}(\xi_t^k[i_t^k],\delta_t^{(K)},\lambda_t^k )\in {\mathbb U}_t^k, \forall t \in T_+^{n-1}$.
\end{theorem}
\begin{proof}
 (1) We show the statement  using  induction and    Theorem~\ref{thm:vp} (alternatively we can use    \cite{kumar-varayia:B1986,bertsekas-shreve1978,vanschuppen2021}).  At time $t=n$,  by  (\ref{dp_1_nnn}),  (\ref{dp_1_nnn_1}), since  
${\cal V}_n^{\gamma^{-k,o}}(i_n^k )$ depends on $i_n^k$ through $\{\xi_{n}^k, \delta_n^{(K)}, \lambda_n^k \}$,  we have ${\cal V}_n^{\gamma^{-k,o}}(i_n^k )={ V}_n^{\gamma^{-k,o}}(\xi_{n}^k, \delta_n^{(K)}, \lambda_n^k )$, as given by    (\ref{dp_1_nnn-mn}),   (\ref{dp_1_nnn_1-mn}).  Suppose at times $j=n, n-1, \ldots, t+1$, ${\cal V}_{j}^{\gamma^{-k,o}}(i_{j}^k )={ V}_{j}^{\gamma^{-k,o}}(\xi_{j}^k, \delta_{j}^{(K)} , \lambda_{j}^k )$ holds.  First, we evaluate the second RHS term in 
(\ref{dp_2_nm})  
using Lemma~\ref{lemma:nested}, 
the DP equations of Theorem~\ref{thm:vp} and    the recursion  of Theorem~\ref{thm:is-pbp}:
\begin{align}
&{\mathbb E}^{ \gamma^{-k,o}} \Big\{{\cal V}_{t+1}^{ \gamma^{-k,o}}(I_{t+1}^k) \Big| i_t^k, u_t^k \Big\} \\
  &={\mathbb E}^{ \gamma^{-k,o}} \Big\{{ V}_{t+1}^{ \gamma^{-k,o}}(\Xi_{t+1}^k, \Delta_{t+1}^{(K)}, \Lambda_{t+1}^k) \Big| i_t^k, u_t^k \Big\}  \\
  &= \int_{ {\mathbb Y}_{t+1}^k \times {\mathbb X}_{t, t+1}   \times {\mathbb L}_{t}^{-k}    }   { V}_{t+1}^{ \gamma^{-k,o}}({\bf T}_{t+1}^{k}[y_{t+1}^{k}, u_t^k,  \gamma_t^{-k,o}(\delta_t^{(K)}, \cdot),  \nonumber \\
  &  \xi_t^k(\cdot)], \delta_t^{(K)}, \lambda_t^k , y_{t+1}^k, U_t^k,    y_{t-T+1}^{-k}, u_{t-T+1}^{-k})  \label{nm-ctg-mn}  \\
 .&  {\bf P}_{t+1}^{\gamma^{k},  \gamma^{-k,o}}(dy_{t+1}^k, d\lambda_{t}^{-k}, dx_{t}, dx_{t+1}\Big| \xi_t^k,   \delta_t^{(K)}, \lambda_t^k,  u_t^k)\Big\}\nonumber 
\end{align}
where (\ref{nm-ctg-mn}) is obtained by using  $i_t^k=\{\delta_t^{(K)}, \lambda_t^k\}$ specifies $\Xi_t^k=\xi_t^k$, 
and where  the conditional PM  in  (\ref{nm-ctg-mn}),  ${\bf P}_{t+1}^{\gamma^{k},  \gamma^{-k,o}}(dy_{t+1}^k, d\lambda_{t}^{-k}, dx_{t}, dx_{t+1}\big| \xi_t^k, \delta_t^{(K)}, \lambda_t^k, u_t^k)$ is derived in Theorem~\ref{thm:vp},  given by   
(\ref{Mar_REC_1-nn-11})-(\ref{dp_1_na-str1}),
i.e., it depends on $\lambda_t^k$ through $\xi_t^k\equiv \xi_t^k[i_t^k]$,  hence  conditionally independent on $\lambda_t^k$.
 Substituting (\ref{nm-ctg-mn}) 
in the RHS of (\ref{dp_2_nm}), we deduce that 
${\cal V}_t^{\gamma^{-k,o}}(i_t^k )={V}_t^{\gamma^{-k,o}}(\xi_t^k, \delta_t^{(K)},\lambda_t^k ), t=1, \ldots, n-1$, which shows  (\ref{dp_1_na_m}) and by equivalence (\ref{dp_2_mn}).  
(2)   At time $t=n$, by  (\ref{dp_1_nnn_1-mn}), we deduce that the optimal strategy is $u_n^{k,o}=\gamma_n^{k,o}(\xi_n^k, \delta_n^{(K)})$.  For any $ t\in  T_+^{n-1}$, by  (\ref{dp_1_na_m}) we deduce that the optimal strategy is $u_t^{k,o}=\gamma_n^{k,o}(\xi_t^k, \delta_t^{(K)}, \lambda_t^k )$.  Hence,  $\gamma^{k,o}(\cdot) \in  {\cal U}_{1,n}^{k,s-sep}$. (2) This  is  shown  in Theorem~\ref{thm:vp}. 
\end{proof}

\ \

\begin{theorem}(Decentralized DP Equations with Private Information State-Verification Theorem)\\
\label{thm:vf-g}
Suppose the value functions $V_t^{\gamma^{-k,o}}(\cdot), \forall t \in T_+^n, \forall k \in {\mathbb Z}_+^K$ 
satisfy the DP  recursions of Theorem~\ref{thm:vp-mn}.

(1) The inequalities hold a.s.  (almost surely), 
\begin{align}
&{V}_n^{\gamma^{-k,o}}(\Xi_n^k,\Delta_n^{(K)}, \Lambda_n^k ) \label{in_1-a} \\ 
&\leq J_{n,n}^{\gamma^{k}, \gamma^{-k,o}}(I_n^k) , \;   \forall \gamma^k \in {\cal U}_{1,n}^k,  \forall k \in {\mathbb Z}_+^K, \nonumber  \\
&{ V}_t^{\gamma^{-k,o}}(\Xi_t^k,\Delta_t^{(K)},\Lambda_t^k )  \label{in_1}\\
& \leq J_{t,n}^{\gamma^{k}, \gamma^{-k,o}}(I_t^k) , \; \forall t \in T_+^{n-1}, \; \forall \gamma^k \in {\cal U}_{1,n}^k, \; \forall k \in {\mathbb Z}_+^K. \nonumber 
\end{align}
where ${J}_{t,n}^{\gamma^{k},\gamma^{-k,o}}(I_t^k)$ is given by  (\ref{opt-ctg-pbp-tc}).

(2)  Given the optimal semi-separated strategies  $\gamma_{1,n}^{-k}=\gamma^{-k,o}\in {\cal U}_{1,n}^{-k, s-sep}$,   let $\gamma^{k,o} \in {\cal U}_{1,n}^{k,s-sep}$ be a semi-separated strategy such that for all $\{\xi_t^k,\delta_t^{(K)},\lambda_{t}^k\}$, strategy $\gamma_t^{k,o}(\xi_t^k,\delta_t^{(K)},\lambda_{t}^k)$ achieves the infimum in the DP equations  (\ref{dp_1_nnn-mn})-(\ref{dp_1_na_m}).   Then $\gamma^{k,o}(\xi_t^k,\delta_t^{(K)},\lambda_t^k )$ is optimal and a.s.,
\begin{align}
&{ V}_n^{\gamma^{-k,o}}(\Xi_n^k,\Delta_t^{(K)}, \Lambda_t^k)  = J_{n,n}^{\gamma^{k,o}, \gamma^{-k,o}}(I_n^k),\; \forall k, \label{in_2-aa-1}\\
&{V}_t^{\gamma^{-k,o}}(\Xi_t^k,\Delta_t^{(K)},\Lambda_{t}^k)  = J_{t,n}^{\gamma^{k,o}, \gamma^{-k,o}}(I_t^k), \;  \forall t \in T_+^{n-1}, \forall k. \label{in_2-aa}
\end{align}
where ${J}_{t,n}^{\gamma^{k,o},\gamma^{-k,o}}(I_t^k)=\mbox{(\ref{opt-ctg-pbp-tc})}$
 with $\gamma^k=
\gamma^{k,o}\in {\cal U}_{1,n}^{k}$. 
\end{theorem}
\begin{proof}  (1) We use induction. By (\ref{opt-ctg-pbp-tc})    for $t=n$, 
\begin{align}
&J_{n,n}^{\gamma^k, \gamma^{-k,o}}(I_n^k)={\mathbb E}^{\gamma^k, \gamma^{-k,o}}   \Big\{   \ell(n, X_n,\gamma_n^k,  \gamma_n^{-k,o}(\Delta_n^{(K)}, \Lambda_n^{-k}))   \Big| I_n^k \Big\} \nonumber \\
&
= \int_{{\mathbb X}_n \times {\mathbb L}_{n}^{-k}}    \ell(n, x_n,\gamma_n^k,  \gamma_n^{-k,o}(\Delta_n^{(K)}, \lambda_n^{-k}))  \Xi_n^k[I_n^k](dx_n, d\lambda_{n}^{-k}) \nonumber \\
&\geq { V}_{n}^{\gamma^{-k,o}}(\Xi_n^k, \Delta_n^{(K)},\Lambda_n^k  ) \label{dp_ft}
\end{align}
hence (\ref{in_1-a}) holds. Suppose  (\ref{in_1-a}) holds for  $t+1, t+2, \ldots, n$. Then
\begin{align}
&{J}_{t,n}^{\gamma^{k},\gamma^{-k,o}}(I_t^k) \sr{(a)}{=} {\mathbb E}^{^{\gamma^k, \gamma^{-k,o}}} \Big\{\ell(t, X_t,\gamma_t^k,  \gamma_t^{-k,o}(I_t^{-k}))  \nonumber\\
&   + {\mathbb E}^{^{\gamma^k, \gamma^{-k,o}}} \Big\{  \sum_{j=t+1}^{n} \ell(j, X_{j},\gamma_{j}^k,  \gamma_{j}^{-k,o}(I_j^{-k}))  \Big| I_{t+1}^k \Big\}    \Big| I_t^k \Big\}  \nonumber
   \\
&\sr{(b)}{\geq}  {\mathbb E}^{\gamma^{k}, \gamma^{-k,o}}  \Big\{  \ell(t, X_t,\gamma_t^k,  \gamma_t^{-k,o}(I_t^{-k})) \nonumber \\
&+{ V}_{t+1}^{\gamma^{-k,o}}( \Xi_{t+1}^k , \Delta_{t+1}^{(K)}, \Lambda_{t+!}^k    ) 
 \Big|I_t^k\Big\}   \label{pd_p1}     \\
  &\sr{(c)}{=}  {\mathbb E}^{\gamma^{k}, \gamma^{-k,o}}  \Big\{ {\mathbb E}^{\gamma^{k}, \gamma^{-k,o}}  \Big\{  \ell(t, X_t,U_t^k,  \gamma_t^{-k,o}(I_t^{-k}))\nonumber \\
  & +{ V}_{t+1}^{\gamma^{-k,o}}( \Xi_{t+1}^k, \Delta_{t+1}^{(K)},     \Lambda_{t+!}^k    )  \Big|I_t^k, U_t^k\Big\}
 \Big|I_t^k\Big\}
 \end{align}
 \begin{align}
 &\sr{(d)}{=} {\mathbb E}^{\gamma^{k}, \gamma^{-k,o}}  \Big\{ {\mathbb E}^{\gamma^{k}, \gamma^{-k,o}}  \Big\{  \ell(t, X_t,U_t^k,  \gamma_t^{-k,o}(I_t^{-k})) \\
  & +{ V}_{t+1}^{\gamma^{-k,o}}( {\bf T}_{t+1}^{k}[Y_{t+1}^{k}, U_{t}^k, \gamma_t^{-k,o}(\Delta_t^{(K)}, \cdot),    \Xi_t^k(\cdot)],   \Delta_t^{(K)}, \Lambda_t^k , \nonumber \\
  & Y_{t+1}^k, U_{t}^k,   Y_{t-T+1}^{-k}, U_{t-T+1}^{-k}  ) \Big|\Xi_t^k,  \Delta_t^{(K)},\Lambda_t^k , U_t^k\Big\} \Big|I_t^k\Big\}\nonumber \\
& \sr{(e)}{\geq}  {\mathbb E}^{\gamma^{k}, \gamma^{-k,o}}  \Big\{{ V}_{t}^{\gamma^{-k,o}}( \Xi_{t}^k,  \Delta_t^{(K)},\Lambda_t^k )  \Big|I_t^k\Big\} \nonumber \\
&   \sr{(f)}{=}{V}_{t}^{\gamma^{-k,o}}( \Xi_{t}^k, \Delta_t^{(K)},\Lambda_t^k  ) \label{pd_p2}
\end{align}
where $(a)$ is by  re-conditioning, $(b)$ is by the hypothesis of induction,  $(c)$ is by  re-conditioning, $(d)$  is by  recursion of Theorem~\ref{thm:is-pbp} and Lemma~\ref{lemma:nested}, i.e.,  $\big\{\Delta_{t+1}^{(K)}, \Lambda_{t+1}^k\big\}  =\big\{    \Delta_{t}^{(K)}, \Lambda_t^k ,Y_{t+1}^k, U_t^k,  Y_{t-T+1}^{-k}, U_{t-T+1}^{-k}\big\}$,  
 and the inner conditional expectation can be expressed  w.r.t. ${\bf P}_{t+1}^{ \gamma^{-k,o}}(dy_{t+1}^k, d\lambda_{t}^{-k}, dx_{t}, dx_{t+1}\big|   \Xi_t^k, \Delta_t^{(K)},\Lambda_t^{k},U_t^k)$   given by (\ref{Mar_REC_1-nn-11})-(\ref{dp_1_na-str1})
independently of  $\Lambda_t^{k}$,  $(e)$ is due to   
  (\ref{dp_1_na_m}), and $(f)$ is because ${ V}_{t}^{\gamma^{-k,o}}( \cdot ) $ is a measurable function of $I_t^k$.  Hence, (\ref{in_1}) holds. (2) We prove (\ref{in_2-aa}) by induction. Suppose a semi-separated strategy $\gamma^{k,o} \in {\cal U}_{1,n}^{k,s-sep}$ achieves the infimum. By (\ref{dp_ft}) we know that  (\ref{in_2-aa-1}) holds for $t=n$. Suppose (\ref{in_2-aa-1})  holds up to  $t+1$. Then with $\gamma^k=\gamma^{k,o} \in {\cal U}_{1,n}^{k,s-sep}$   all inequalities in  (\ref{pd_p1})-(\ref{pd_p2})  become equalities (by the induction hypothesis and the fact that $u_t^{k,o}=\gamma_t^{k,o}(\xi_t^k, \delta_t^{(K)},\lambda_t^k )$ achieves the infimum). Hence,  (\ref{in_2-aa}) holds for  all $t \in T_+^{n-1}$. For $t=1$ we have ${ V}_{1}^{\gamma^{-k,o}}( \Xi_{1}^k, \Delta_1^{(K)},\Lambda_1^k   ) ={J}_{1,n}^{\gamma^{k,o},\gamma^{-k,o}}(I_1^k)$, and by  taking expectation on both sides,  
\begin{align*}
&{\mathbb E}^{\gamma^{k,o},\gamma^{-k,o}}\big\{V_{1}^{\gamma^{-k,o}}( \Xi_{1}^k, \Delta_1^{(K)},\Lambda_1^k     )\big\}\\
&={\mathbb E}^{\gamma^{k,o}, \gamma^{-k,o}}\big\{V_{1}^{\gamma^{-k,o}}( \Xi_{1}^k, \Delta_1^{(K)},\Lambda_1^k  )\big\} \nonumber \\
&={\mathbb E}^{\gamma^{k,o},\gamma^{-k,o}}\big\{{J}_{1,n}^{\gamma^{k,o},\gamma^{-k,o}}(I_1^k)\big\}
={J}_{n}(\gamma^{k,o},\gamma^{-k,o}).
\end{align*}
 Finally, for any other $\gamma^{k} \in {\cal U}_{1,n}^{k}$, which does not achieve the infimum of the DP equations, evaluating (\ref{in_1}) at  $t=1$ and taking expectations using the fact that the value function does not depend on $\gamma^k$,  we have $\forall \gamma^k \in {\cal U}_{1,n}^{k}$, that ${\mathbb E}^{\gamma^{k}, \gamma^{-k,o}}\big\{V_{1}^{\gamma^{-k,o}}( \Xi_{1}^k, \Delta_1^{(K)},\Lambda_1^k      )\big\} \leq  {J}_{n}(\gamma^{k},\gamma^{-k,o})$.
 This completes the proof.
\end{proof}

\ \

\begin{remark} (On Theorem~\ref{thm:vp-mn} and Theorem~\ref{thm:vf-g})\\
\indent  In Theorems~\ref{thm:vp-mn}.(2), \ref{thm:vf-g}.(2),  ${\cal G}_t^k\tri \{\Xi_t^k,\Delta_t^{(K)},\Lambda_t^k \}$ is a sufficient statistic for
$\gamma^{k,o}\in {\cal U}_{1,n}^{k,s-sep}$,  but  expansive w.r.t.  $\Delta_t^{(K)},\forall  t$. We address this issue in Section~\ref{sect:DP-IS-SEP}. 

\end{remark}

\subsection{DP Equations, Private and Centralized  Information States, Separated  and Information state  Strategies}
\label{sect:DP-IS-SEP}
The subsequent Theorems~\ref{thm:vp-mn-sep},  \ref{thm:vp-mn-is},   settle the long standing open problem, first discussed by Witsenhausen in \cite[Assertion 8, pp.1562]{witsenhausen1971}. We make use of the  centralized \^a posteriori PMs, Theorem~\ref{thm:is-cen},  to derive    2   additional simplified  DP equations,  using separable strategies    ${\cal U}_{1,n}^{k,sep}$ and information state strategies ${\cal U}_{1,n}^{k,is}$,  which are not expansive w.r.t.  $\Delta_t^{(K)},\forall  t$. 

This  part of  the paper  is a generalization  of the classical centralized DP approach  of POMDP,  where a  centralized \^a posteriori PM of nonlinear filtering  is used  \cite{striebel1965,bertsekas-shreve1978,kumar-varayia:B1986,paris-borkar-emmanuel-ghoshi-marcus93,vanschuppen2021}
 and \cite[Theorem~7.1, pp.85-86]{kumar-varayia:B1986}.

 \ \
 
 \begin{definition}(Decentralized PbP Value Functions-Private  and Centralize  Information  States)\\
\label{def:sep}
Consider  separated strategies  $\gamma^k \in {\cal U}_{1,n}^{k,sep}\subseteq {\cal U}_{1,n}^{k,s-sep}\subseteq{\cal U}_{1,n}^{k}$,  $\gamma_t^k=\gamma_t^{k}(\Xi_t^k[I_t^k], \Theta_t[\Delta_{t}^{(K)}], \Lambda_t^k),I_t^k=\{\Theta_t, \Lambda_t^k\},  \forall t, k$. 

For each  $k\in {\mathbb Z}_+^K$, define the value function,    $V_t^{\gamma^{k,o}, \gamma^{-k,o},sep}(\cdot) :  {\cal M}({\mathbb X}_t \times {\mathbb L}_{t}^{-k}) \times{\cal M}({\mathbb X}_{t-T}) \times {\mathbb L}_t^k    \rar (-\infty,\infty)$
of  strategy $\gamma^k \in {\cal U}_{1,n}^{k,sep}$,
 conditional  on 
 $  \xi_t^k[i_t^k]\tri {\bf P}_{t}^{{ \gamma^{-k,o}}} (dx_{t}, d\lambda_{t}^{-k}\big|I_{t}^k=i_t^k) $,   $ \theta_{t}[\delta_t^{(K)}] \tri {\bf P}_t^{\gamma^{(K)}} (dx_{t-T}\big|\Delta_t^{(K)}=\delta_t^{(K)})$, and $\Lambda_t^k=\lambda_t^k$,    by
\begin{align}
&V_t^{\gamma^o, \gamma^{-k,o},sep}(\xi_t^k,\theta_t,\lambda_t^k )   \hst \hst   \forall t \in T_+^n, \forall k \in {\mathbb Z}_+^K       \nonumber   \\
&  \tri  \inf_{  \gamma^{k}\in {\cal U}_{t,n}^{k,sep}}   {\mathbb E}^{^{\gamma^k, \gamma^{-k,o}}} \Big\{
 \sum_{j=t}^{n} \ell(j, X_j,\gamma_j^k,  \gamma_j^{-k,o}) \nonumber \\
 &  \Big| \Xi_t^k=\xi_t^k ,\Theta_t=\theta_t,\Lambda_t^k =\lambda_t^k ,  \gamma_t^k  \Big\}\label{opt-ctg}  \\
&= \inf_{  \gamma^{k}\in {\cal U}_{t,n}^k}   {\mathbb E}^{^{\gamma^k, \gamma^{-k,o}}} \Big\{
 \sum_{j=t}^{n}  \int_{{\mathbb X}_j \times {\mathbb L}_{j}^{-k}}   \ell(j, X_j,\gamma_j^k,  \gamma_j^{-k,o})\nonumber \\
&. \Xi_j^k[I_j^k](dx_j, d\lambda_{j}^{-k})\Big| \Xi_t^k=\xi_t^k,  \Theta_t= \theta_t, \Lambda_t^k=\lambda_t^k , \gamma_t^k\Big\}. \label{opt-ctg-n} 
\end{align}
\end{definition}


\ \

\begin{theorem}(DP Equations with Private and Centralized  Information States-Necessary Conditions and Verification)\\
\label{thm:vp-mn-sep}
Suppose   Assumptions~\ref{ass-1} hold.  
 Consider the value function $V_t^{\gamma^{k,o}, \gamma^{-k,o},sep}(\xi_t^k,\theta_t,\lambda_t^k ), \forall t \in T_{+}^n$ of Definition~ \ref{def:sep} 

 (1) Suppose a decentralized team equilibrium  $\gamma^{k,o} \in {\cal U}_{t,n}^{k,sep}, \forall k$ exists. Then  $ V_t^{\gamma^{k,o}, \gamma^{-k,o},sep}(\cdot )=V_t^{ \gamma^{-k,o},sep}(\cdot), \forall \gamma^k \in {\cal U}_{1,n}^{k,sep}, \forall t \in T_{+}^n$ and 
   satisfies    the following  DP   equations.  
\begin{align}
&{ V}_n^{\gamma^{-k,o},sep}(\xi_n^k,\theta_n, \lambda_n^k  )=\inf_{u_n^k \in {\mathbb U}_n^k}{\mathbb E}^{ \gamma^{-k,o}}   \Big\{ \ell(n, X_n,u_n^k,   \nonumber  \\
&  \gamma_n^{-k,o}(\Xi_n^{-k}, \theta_n, \Lambda_n^{-k}))   \Big| \xi_n^k,\theta_n, \lambda_n^k  , u_n^k \Big\}, \forall k \label{dp_1_nnn-mn-is-s} \\
&= \inf_{u_n^k \in {\mathbb U}_n^k}\int_{{\mathbb X}_n  \times {\mathbb L}_{n}^{-k} }     \ell(n, x_n,u_n^k,  \gamma_n^{-k,o}(\xi_n^{-k}, \theta_n, \lambda_n^{-k}))\nonumber \\
&\hst . \xi_{n}^{k}[i_t^k] (dx_{n}, d\lambda_n^{-k}),        \label{dp_1_nnn_1-mn-is-s} \\
&{ V}_t^{\gamma^{-k,o},sep}(\xi_t^k,\theta_t,\lambda_t^k  )   \hst \hso     \forall t \in T_+^{n-1}, \; \forall k \nonumber
   \\
& = \inf_{u_t^k \in {\mathbb U}_t^k}  {\mathbb E}^{ \gamma^{-k,o}}   \Big\{ \ell(t, X_t,u_t^k,  \gamma_t^{-k,o}(\Xi_t^{-k}, \theta_t, \Lambda_t^{-k}))  \nonumber \\
&+{ V}_{t+1}^{\gamma^{-k,o}, sep}(\Xi_{t+1}^k, \Theta_{t+1}, \Lambda_{t+1}^k )\Big| \xi_t^k,\theta_t, \lambda_t^k , u_t^k, u_{t-T}^{(K)} \Big\} \label{dp_2_mn-is-s}  \\
 &=  \inf_{u_t^k \in {\mathbb U}_t^k}  \Big\{  \int_{{\mathbb X}_t \times {\mathbb L}_{t}^{-k}}    \ell(t, x_t,u_t^k,  \gamma_t^{-k,o}(\xi_t^{-k}, \theta_t, \lambda_t^{-k}))  \nonumber \\
 &.  \xi_{t}^k[i_t^k] (dx_{t}, d\lambda_{t}^{-k})  +  \int_{ {\mathbb Y}_{t+1}^k \times {\mathbb X}_{t, t+1}   \times {\mathbb L}_{t}^{-k}    }   {V}_{t+1}^{\gamma^{-k,o},sep}({\bf T}_{t+1}^{k}[\nonumber  \\
 &y_{t+1}^{k},  u_t^k, \gamma_t^{-k,o}(\xi_t^{-k}, \theta_t, \cdot),    \xi_t^k(\cdot)], {\bf T}_{t+1}^{(K)}[y_{t-T+1}^{k},    y_{t-T+1}^{-k}, \nonumber  \\
 &   u_{t-T}^{(K)}, \theta_{t} (\cdot)],     y_{t+1}^k, u_t^k,  y_{t-T+2,t}^{k},  u_{t-T+2,t-1}^{k})  \label{dp_1_na_m-is-s} \\
&. {\bf P}_{t+1}^{ \gamma^{-k,o}}(dy_{t+1}^{k}, d\lambda_{t}^{-k}, dx_{t}, dx_{t+1}     \Big|   \xi_t^k, \theta_t, \lambda_t^k , u_t^k, u_{t-T}^{(K)}) \Big\} \nonumber 
\end{align}
where  
$\lambda_t^{j}=\{y_{t-T+1,t}^{j},u_{t-T+1,t-1}^{j}\},\forall  j$,  and 
 the   conditional  PM of the last right hand side term is
\begin{align}
&{\bf P}_{t+1}^{ \gamma^{-k,o}}( dy_{t+1}^{k}, d\lambda_t^{-k},dx_{t}, dx_{t+1} \big|\xi_t^k, \theta_t, \lambda_t^k,u_t^k, u_{t-T}^{(K)})   \nonumber 
\\
&= Q_{t+1}^{k}(dy_{t+1}^{k}\big|    x_{t+1},  u_t^k, \gamma_t^{-k,o}(\xi_t^{-k}, \theta_t, \lambda_t^{-k}))\label{Mar_REC_1-mn-is-s}  \\
&.{ S}_{t+1}(dx_{t+1}\big|   x_t, u_t^k,\gamma_t^{-k,o}(\xi_t^{-k},\theta_t, \lambda_t^{-k}))\xi_{t}^k[i_t^k](dx_{t}, d\lambda_{t}^{-k})
\nonumber \\
&\equiv  {\bf P}_{t+1}^{{ \gamma^{-k,o}}}(dy_{t+1}^{k}, d\lambda_{t}^{-k},dx_t, dx_{t+1} \big|  \xi_t^k, \theta_t, u_t^k) . \label{Mar_REC_1-mn-is-ss} 
\end{align}

(2) For each $ k \in {\mathbb Z}_+^K$ suppose   the DP equations (\ref{dp_1_nnn-mn-is-s})-(\ref{dp_1_na_m-is-s}) hold. Then 
the following inequalities hold, 
\begin{align}
&{V}_n^{\gamma^{-k,o},sep}(\Xi_n^k,\Theta_n, \Lambda_n^k ) \label{in_1-a-is-s} \\ 
&\leq J_{n,n}^{\gamma^{k}, \gamma^{-k,o}}(I_n^k) , \;   \forall \gamma^k \in {\cal U}_{1,n}^{k,sep},  \forall k \in {\mathbb Z}_+^K, \nonumber  \\
&{ V}_t^{\gamma^{-k,o},sep}(\Xi_t^k,\Theta_t,\Lambda_t^k )  \label{in_1-is-s}\\
& \leq J_{t,n}^{\gamma^{k}, \gamma^{-k,o}}(I_t^k) , \; \forall t \in T_+^{n-1}, \; \forall \gamma^k \in {\cal U}_{1,n}^{k,sep}, \; \forall k \in {\mathbb Z}_+^K. \nonumber 
\end{align}
where ${J}_{t,n}^{\gamma^{k},\gamma^{-k,o}}(I_t^k)$ is defined  by  (\ref{opt-ctg-pbp-tc})  w.r.t. $\{\gamma^{k},\gamma^{-k,o}\} \in {\cal U}_{1,n}^{k,sep}\times {\cal U}_{1,n}^{-k,sep}$.

(3)  Given the optimal separated strategies  $\gamma_{1,n}^{-k}=\gamma^{-k,o}\in {\cal U}_{1,n}^{-k, sep}$,   let $\gamma^{k,o} \in {\cal U}_{1,n}^{k,sep}$ be a separated strategy such that for all $\{\xi_t^k,\theta_t,\lambda_{t}^k\}$, strategy $\gamma_t^{k,o}(\xi_t^k,\theta_t^{(K)},\lambda_{t}^k)$ achieves the infimum in DP eqns   (\ref{dp_1_nnn-mn-is-s})-(\ref{dp_1_na_m-is-s}) .   Then $\gamma^{k,o}(\xi_t^k,\theta_t,\lambda_t^k )$ is optimal and $\forall k \in {\mathbb Z}_+^K$,
\begin{align}
&{ V}_n^{\gamma^{-k,o}, sep}(\Xi_n^k,\Theta_t, \Lambda_t^k)  = J_{n,n}^{\gamma^{k,o}, \gamma^{-k,o}}(I_n^k)-a.s., \label{in_2-a}\\
&{V}_t^{\gamma^{-k,o},sep}(\Xi_t^k,\Theta_t,\Lambda_{t}^k) \nonumber \\
 & = J_{t,n}^{\gamma^{k,o}, \gamma^{-k,o}}(I_t^k)-a.s., \forall t \in T_+^{n-1}. \label{in_2}
\end{align}
where ${J}_{t,n}^{\gamma^{k,o},\gamma^{-k,o}}(I_t^k)$ is defined  as in (2). 

(4) If $\big\{\Delta_t^{(K)}| t \in T_+^{n-1}\big\}$ and $\big\{\Theta_t[\Delta_t^{(K)}]| t \in T_+^{n-1}\big\}$ generate the same information,  then  ${\cal U}_{1,n}^{-k, s-sep}={\cal U}_{1,n}^{-k, sep}$   and     ${ V}_t^{\gamma^{-k,o}}(\Xi_t^k,\Delta_t^{(K)},\Lambda_t^k )={ V}_t^{\gamma^{-k,o},sep}(\Xi_t^k,\Theta_t,\Lambda_t^k )-a.s., \forall t$.
\end{theorem}
\begin{proof} (1) The proof  is similar to  Theorem~\ref{thm:vp-mn}.(1).
%
%
%
(2), (3) The proof  is  similar to  Theorem~\ref{thm:vf-g}. 
  (4)  By the hypothesis, it follows  that   Theorem~\ref{thm:vp-mn},    Theorem~\ref{thm:vf-g}  hold, with    $\Delta_t^{(K)}$ replaced by   $\Theta_t[\Delta_t^{(K)}],\forall  t \in T_+^{n-1}$.
\end{proof} 

\ \

In Theorem~\ref{thm:vp-mn-is},  we  consider the centralized  \^a posteriori PM $ \big\{ \Pi_t^{\gamma^{(K)}}[\delta_{t}^{(K)}] \big| t \in T_+^n\big\}$  of  Theorem~\ref{thm:is-cen}.(2), and  we show
 $V_t^{\gamma^{-k,o},sep}(\xi_t^k,\pi_t^{\gamma^{(K)}},\lambda_t^k  ) =V_t^{\gamma^{-k,o},is}(\xi_t^k,\pi_t^{\gamma^{(K)}})
  $, i.e.,  the DP equations
  depend on $\lambda_t^k$ through $\xi_t^k,   \forall (\xi_t^k,\pi_t^{\gamma^{(K)}},\lambda_t^k),  \forall t$.

\ \

\begin{theorem}(DP Equations with Private and Centralized  Information State-Information State Depended Strategies)\\
\label{thm:vp-mn-is}
Suppose   Assumptions~\ref{ass-1} hold.  \\
 Consider the value function of Definition~ \ref{def:sep}, with strategies $ \gamma^k \in {\cal U}_{1,n}^{k,sep}$, corresponding to  $\gamma_t^k=\gamma_t^{k}(\Xi_t^k[I_t^k], \Pi_t^{\gamma^{(K)}}[\Delta_{t}^{(K)}], \Lambda_t^k),I_t^k=\{\Pi_t^{\gamma^{(K)}}, \Lambda_t^k\},  \forall (t, k)$ (i.e., 
  $ \big\{ \Theta_t[\delta_{t}^{(K)}] \big| t \in T_+^n\big\}$  is replaced by $ \big\{ \Pi_t^{\gamma^{(K)}}[\delta_{t}^{(K)}] \big| t \in T_+^n\big\}$), denoted again by  
  $V_t^{\gamma^{k,o},\gamma^{-k,o},sep}(\xi_t^k,\pi_t^{\gamma^{(K)}},\lambda_t^k  ), \forall t \in T_{+}^n$. The following hold.  
  
  (1)  Suppose a decentralized team equilibrium $\gamma^{k,o} \in {\cal U}_{1, n}^{k,sep}, \forall k$ exists. Then 
   the  identity holds, 
\begin{align}   
  &  V_t^{\gamma^{k,o},\gamma^{-k,o},sep}(\xi_t^k,\theta_t^{\gamma^{(K)}},\lambda_t^k  ) =V_t^{\gamma^{-k,o},is}(\xi_t^k,\pi_t^{\gamma^{(K)}}),\nonumber  \\
  & \forall (\xi_t^k,\pi_t^{\gamma^{(K)}},\lambda_t^k  ), \;  \forall t \in T_{+}^n.
\end{align}    
 Moreover,  
  $V_t^{\gamma^{-k,o},is}(\xi_t^k,\pi_t^{\gamma^{(K)}})$ 
 satisfies the DP equations, 
\begin{align}
&{ V}_n^{\gamma^{-k,o},is}(\xi_n^k,\pi_n^{\gamma^{(K)}})= \inf_{u_n^k \in {\mathbb U}_n^k}\int_{{\mathbb X}_n  \times {\mathbb L}_{n}^{-k} } \ell(n, x_n, \nonumber   \\
& u_n^k,   \gamma_n^{-k,o}(\xi_n^{-k},\pi_n^{\gamma^{(K)}}))  \xi_{n}^{k} [i_n^k](dx_{n}, d\lambda_n^{-k}),      \label{dp_1_nnn_1-mn-is} \\
&{ V}_t^{\gamma^{-k,o},is}(\xi_t^k,\pi_t^{\gamma^{(K)}} )  \label{dp_2_mn-is}   \\
 &=  \inf_{u_t^k \in {\mathbb U}_t^k}  \Big\{  \int_{{\mathbb X}_t \times {\mathbb L}_{t}^{-k}}    \ell(t, x_t,u_t^k,  \gamma_t^{-k,o}(\xi_t^{-k},\pi_t^{\gamma^{(K)}}))  \nonumber \\
 &.  \xi_{t}^k[i_t^k] (dx_{t}, d\lambda_{t}^{-k})  +  \int_{ {\mathbb Y}_{t+1}^k \times {\mathbb X}_{t, t+1}   \times {\mathbb L}_{t}^{-k}    } \Big[{V}_{t+1}^{\gamma^{-k,o},is}({\bf T}_{t+1}^{k}[   \nonumber  \\
 &y_{t+1}^{k}, u_t^k,\gamma_t^{-k,o}(\xi_t^{-k}, \pi_t^{\gamma^{(K)}}),    \xi_t^k(\cdot)], {\bf T}_{t+1}^{{(K)}}[y_{t+1}^{(K)}, u_{t}^{k}, \nonumber  \\ 
 & \gamma_{t}^{-k,o}(\xi_t^{-k},\pi_t^{\gamma^{(K)}}), \pi_{t}^{\gamma^{(K)}} (\cdot)]) \nonumber  \\
 &. {\bf P}_{t+1}^{ \gamma^{-k,o}}(dy_{t+1}^{(K)}, d\lambda_{t}^{-k}, dx_{t}, dx_{t+1}\big|   \xi_t^k, \pi_t^{\gamma^{(K)}}, u_t^k) \Big]  \Big\} \label{dp_1_na_m-is}
\end{align}
where  
 the   conditional  PM of the last RHS term is
\begin{align}
&{\bf P}_{t+1}^{ \gamma^{-k,o}}( dy_{t+1}^{(K)}, d\lambda_t^{-k},dx_{t}, dx_{t+1} \big|\xi_t^k, \pi_t^{\gamma^{(K)}},u_t^k)  \label{Mar_REC_1-mn-aa-is} \\
&= Q_{t+1}^{(K)}(dy_{t+1}^{(K)}\big|    x_{t+1},  u_t^k, \gamma_t^{-k,o}(\xi_t^{-k},\pi_t^{\gamma^{(K)}}))\label{Mar_REC_1-mn-is}  \\
&.{ S}_{t+1}(dx_{t+1}\big|   x_t, u_t^k,\gamma_t^{-k,o}(\xi_t^{-k},\pi_t^{\gamma^{(K)}}))\xi_{t}^k[i_t^k](dx_{t}, d\lambda_{t}^{-k}).
\nonumber 
\end{align}
Moreover,  if  the infimum in the DP equations exists, then the optimal strategy occurs in the set of  strategies, $\gamma^{k,o} \in  {\cal U}_{1,n}^{k, is}$, i..e., ${\cal G}_t^k\tri \{\Xi_t^k,\Pi_t^{\gamma^{(K)}}\}$ is a sufficient statistic for $U_t^k$, $\forall t$. In particular, 
$u_t^{k,o}=\gamma_t^{k,o}(\xi_t^k,\pi_t^{\gamma^{(K)}} )\in {\mathbb U}_t^k, \forall t \in T_+^{n}$.

(2)   The statements of Theorem~\ref{thm:vp-mn-sep}.(2), (3) w.r.t.    ${ V}_t^{\gamma^{-k,o},is}(\xi_t^k,\pi_t^{(K)} ), \forall (t, k)  $  hold. 
\end{theorem}
\begin{proof} 
 (1)    We show the statement  using  induction. By Definition~ \ref{def:sep} with $ \theta_t[\delta_{t}^{(K)}]$  replaced by $  \pi_t^{\gamma^{(K)}}[\delta_{t}^{(K)}], \forall t $, we obtain that  the  DP equation of  $V_t^{\gamma^{k,o},\gamma^{-k,o},sep}(\xi_n^k,\pi_t^{\gamma^{(K)}},\lambda_t^k  ), \forall t$  satisfies (\ref{dp_1_nnn_1-mn-is-s}), (\ref{dp_2_mn-is-s}).   At  time $t=n$, from  $V_n^{\gamma^{k,o},\gamma^{-k,o},sep}(\xi_n^k,\pi_n^{\gamma^{(K)}},\lambda_n^k  )$ given by the analog of   (\ref{dp_1_nnn_1-mn-is-s}),  we obtain the  optimal strategy $u_n^{k,o}=\gamma_n^{k,o}(\xi_n^k, \pi_n^{\gamma^{(K)}})$, i.e., $\gamma_n^{k,o}(\cdot) \in  {\cal U}_{n}^{k, is}$, which then implies  $V_n^{\gamma^{k,o},\gamma^{-k,o},sep}(\xi_n^k,\pi_n^{\gamma^{(K)}},\lambda_n^k  )=V_n^{\gamma^{-k,o},is}(\xi_n^k,\pi_n^{\gamma^{(K)}}), \forall \lambda_n^k$. Suppose at times $j=n, n-1, \ldots, t+1$, ${ V}_{j}^{\gamma^{,o},\gamma^{-k,o},sep}(\xi_{j}^k, \pi_{j}^{\gamma^{(K)}}, \lambda_{j}^k )={ V}_{j}^{\gamma^{-k,o},is}(\xi_{j}^k, \pi_{j}^{\gamma^{(K)}}), \forall (\xi_j^k,\theta_j^{\gamma^{(K)}},\lambda_j^k )$ holds. 
    Evaluating  the second RHS term in 
(\ref{dp_2_mn-is-s}),     ${ V}_{t+1}^{\gamma^{-k,o}, is}(\Xi_{t+1}^k, \Pi_{t+1}^{\gamma^{(K)}})$, using $\Pi_{t+1}^{\gamma^{(K)}}$  of Theorem~\ref{thm:is-cen}.(2),   we deduce the optimal strategy is $\gamma_t^{k,o}(\xi_t^k,\pi_t^{\gamma^{(K)}} )\in {\mathbb U}_t^k$,  
$\gamma_t^{k,o}(\cdot)$ depends on $\lambda_t^k$ through  $\xi_t^k$, hence
 and $\gamma_t^{k,o}(\cdot) \in {\cal U}_{1,n}^{k,is}, \forall t \in T_+^{n}$. Since this holds $\forall k$ 
 we obtain the RHS of (\ref{dp_1_na_m-is}), and hence 
  the LHS of the DP equation satisfies $V_t^{\gamma^{k,o},\gamma^{-k,o},sep}(\xi_t^k,\pi_t^{\gamma^{(K)}},\lambda_t^k  ) =V_t^{\gamma^{-k,o},is}(\xi_t^k,\pi_t^{\gamma^{(K)}}),
 \forall (\xi_t^k,\pi_t^{\gamma^{(K)}},\lambda_t^k )$. 
 This shows (1). (2) This follows from the statements of Theorem~\ref{thm:vp-mn-sep}. 
\end{proof}

%

\ \

\begin{remark} (Comparing Theorems~\ref{thm:vp-mn-sep},\ref{thm:vp-mn-is})\\
\label{rem:sep-is-com} 
Theorem~\ref{thm:vp-mn-sep}  uses the  centralized  \^a posteriori PM,  
 $ \theta_{t}[\delta_t^{(K)}] \tri {\bf P}_t^{\gamma^{(K)}} (dx_{t-T}\big|\delta_t^{(K)})={\bf P}_t(dx_{t-T}\big|\delta_t^{(K)}),\forall  t, \forall \gamma^{(K)}$, i.e., which does not   depend on $\gamma^{(K)}$.
This  is  contrary to  Theorem~\ref{thm:vp-mn-is} that uses   $ \pi_{t}^{\gamma^{(K)}}[\delta_t^{(K)}] \tri {\bf P}_t^{\gamma^{(K)}} (dx_{t}, d\lambda^{(K)}\big|\delta_t^{(K)}),$
  $ \pi_{t+1}^{\gamma^{(K)}}={\bf T}_{t+1}^{(K)}[y_{t}^{(K)}, \gamma_t^{(K)}, 
  \pi_{t}^{\gamma^{(K)}} (\cdot)]$, 
  which depends    on $y_t^{(K)}$ and  the strategies   $\gamma^{(K)}$, by Theorem~\ref{thm:is-cen}.(2). 
\end{remark}

\ \

\begin{remark} (Without Sharing  with Perfect Recall)\\
\label{rem:1-ws} 
 Consider  
 strategies without sharing
   i.e., $I_t^k=  \sr{\circ}{I}_t^k \tri  \big\{Y_{1,t}^{k}, U_{1,t-1}^{k}\big\}\equiv \Lambda_t^k$ and  
  $\sr{\circ}{\gamma}^k \in \sr{\circ}{\cal  U}_{1,n}^k$,    $U_t^k=\sr{\circ}\gamma_t^k({\Lambda}_t^k),   \forall (t,k)$.
 Let
$\sr{\circ}{\Xi}_t^k[\Lambda_t^k]\tri  {\bf P}_t^{ \sr{\circ}{\gamma}^{-k,o}} (dx_t, d {\lambda}_{t}^{-k}\big|{\Lambda}_t^k),   \forall (t,k) $ satisfy the recursions of Theorem~\ref{thm:is-pbp}, with $\Delta_t^{(K)}=\{\emptyset\}, I_t^k=\sr{\circ}{I}_t^k=\Lambda_t^k$.  Then the equality holds, 
\begin{align}
V_t^{ \gamma^{-k,o}}(\Xi_t^k, \Delta_t^{(K)}, \Lambda_t^k) =\sr{\circ}{V}_t^{ \sr{\circ}{\gamma}^{-k,o}}(\sr{\circ}{\Xi}_t^k),   \hso \forall (t, k)
\end{align}
where $\sr{\circ}{V}_t^{ \sr{\circ}{\gamma}^{-k,o}}(\cdot)$ satisfies degenerate DP equations. In this special case, strategy $\sr{\circ}\gamma_t^k({\Lambda}_t^k)=\sr{\circ}\gamma_t^k({\Xi}_t^k)$, hence only  the private information state $\sr{\circ}{\Xi}_t^k$ is used. 
\end{remark}

Finally, we note that  the information states and DP equations in  
\cite[Theorem~5, Theorem~7]{malikopoulosIEEEAC2023} are incorrect  (see  \cite{charalambous-guvercin-djouadi:arxiv2026-05-25}).

\ \



\section{The Single Cost-to-Go   Based on Common Information}
\label{sec:lr}
Past   studies \cite{witsenhausen1971,kurtaran-sivan1973,sandell-athans1974,kurtaran1975,yoshikawa1975,varaiya-walrand1978,
bagchi-basar1980,hsu-marcus:1982,aicardi-davoli-minciardi1987,nayyar-mahajan-teneketzis2011,nayyar-mahajan-teneketzis2013}, 
considered  the single cost-to-go   conditioned on 
$\Delta_t^{(K)}=\delta_t^{(K)}$, 
\begin{align}
&{v}_{t}^{\gamma^{(K),o}}(\delta_t^{(K)}) \tri  \inf_{\gamma_{t, n}^{(K)}}  {\mathbb E}^{\gamma^{(K)}} \Big\{ \sum_{j=t}^{n} \ell(j, X_j,\gamma_j^1(\Delta_j^{(K)}, \Lambda_j^{1}),   \nonumber \\
&\ldots, \gamma_j^K(\Delta_j^{(K)}, \Lambda_j^{K})) \Big|\Delta_t^{(K)}= \delta_t^{(K)} \Big\},   \forall t   . \label{v-w-ctg}
\end{align}
Early  studies \cite{witsenhausen1971,kurtaran-sivan1973,sandell-athans1974,kurtaran1975,yoshikawa1975,varaiya-walrand1978,
bagchi-basar1980,hsu-marcus:1982,aicardi-davoli-minciardi1987}, made use of  \cite[Assertion~8, pp.1562]{witsenhausen1971}, i.e.,  {\it restricting  the strategies to}
%
 \begin{align}
&\gamma_t^k(\delta_t^{(K)}, \lambda_t^{k})=\phi_t^{k,W}(F_t[\delta_t^{(K)}]  , \lambda_t^{k}), \hso  \forall t,k,    
\label{wit-a1} 
 \\
& F_t[\delta_t^{(K)}]  \tri   {\bf P}_{t}^{\gamma^{(K)}}(dx_{t-T}\big| \delta_t^{(K)} ) \equiv  \theta_t[\delta_t^{(K)}].
  \label{wit-a2}
 \end{align}  
Since  Varaiya and Walrand \cite{varaiya-walrand1978} proved  strategies (\ref{wit-a1}) are not optimal, other strategies are proposed in the literature. \\
For model (\ref{NDM-3-w}), (\ref{NDM-4-w}), \cite{kurtaran1979} considered  strategies, 
\begin{align}
\gamma_t^k=\phi_t^k(y_{t-T+1, t}^k, {\bf P}_t^{\gamma^{(K)}}(dx_{t-T}, du_{t-T+1, t-1}^{(K)}\big| \delta_t^{(K)})) \label{k-1}
\end{align}
$ \forall t, k$, while \cite{nayyar-mahajan-teneketzis2011}  considered   various strategies including, 
\begin{align}
\gamma_t^k=\phi_t^k(\lambda_t^{k}, {\bf P}_t^{\gamma^{(K)}}(dx_{t-1}, d\lambda_t^{(K)}\big| \delta_t^{(K)})).\label{n-a-t-1}
\end{align}
 In  \cite{nayyar-mahajan-teneketzis2011} it is acknowledged that their results    cannot be derived from those  in \cite{kurtaran1979}, and vice-versa. Furthermore, \cite{nayyar-mahajan-teneketzis2011} also acknowledged  that their single DP approach based on   (\ref{v-w-ctg}), and strategies (\ref{n-a-t-1}),  do not satisfy fundamental Properties \#1,  \#2, i.e., the value function and conditional PM ${\bf P}_t^{\gamma^{(K)}}(dx_{t-1},d \lambda_t^{(K)}\big| \delta_t^{(K)})$ depend on the strategies and not the actions, 
    \cite[pp.1607]{nayyar-mahajan-teneketzis2011}, i.e., their information states  ``{\it do not have the separated
nature of centralized stochastic control}''.  \\
Moreover, \cite{nayyar-mahajan-teneketzis2011} also acknowledged    the   major disadvantage  of the single cost-to-go  (\ref{v-w-ctg}), re-confirming  Witsenhausen \cite{witsenhausen1971} (see paragraph below Assertion 8):  if there is no  sharing, i.e., 
\begin{align}
&\Delta_t^{(K)}=\{\emptyset\},  \;{I}_t^k=\sr{\circ}{I}_t^{k}\tri  \big\{Y_{1,t}^k, U_{1,t-1}^k\big\},   \forall t, \; \mbox{and }\\
&\mbox{(\ref{v-w-ctg})  reduces to} \:  {v}_{t}^{\gamma^{(K),o}}(\Delta_t^{(K)})={v}_{t}^{\gamma^{(K),o}}(\{\emptyset\}), \forall t
\end{align}
hence  the cost-to-go  is {\it vacuous}.

The reader may verify 
the limitations listed in Section~\ref{sect:lit}.1), concerning the single cost-to-go ${v}_{t}^{\gamma^{(K),o}}(\delta_t^{(K)})$.


On the other hand, 
 the   DP equations   of  Theorem~\ref{thm:vp} enjoy properties analogous to those of   classical  centralized DP equations of POMDPs, do not suffer from the limitations of  (\ref{v-w-ctg}), and  give  rise to additional DP equations.

Theorem~\ref{thm:vp-mn-sep}, \ref{thm:vp-mn-is} present   the  most efficient DP equations ${V}_t^{\gamma^{-k,o},sep}(\xi_t^k,\theta_t,\lambda_{t}^k)$ for strategies $\gamma^k\in {\cal U}_{1,n}^{k,sep}$, in the sense that  the data $I_t^k=i_t^k $ are  compressed into  private, centralized information states and limited  memory data $\{\xi_t^k, \theta_t, \lambda_{t}^k\}$,  while   these are sufficient statistics for    $\gamma_t^k, \forall (t, k)$.

Finally, we emphasize that the {\it semi-separated strategies},   {\it separated strategies}, and {\it information state strategies}, can also be used to derive DP equations for the single cost-to-go  (\ref{v-w-ctg}). In particular,  the analysis in  Section~\ref{discrete}  supports the claim that  Witsenhausen's \cite[Assertion~8, pp.1562]{witsenhausen1971} should be replaced by strategies  $\gamma^{k} \in {\cal U}_{1,n}^{k,sep}$, i.e., $\gamma_t^k=\gamma_t^k(\Xi_t^k, \Theta_t, \Lambda_t^{k})$, i.e., $\Theta_t=F_t[\Delta^{(K)}]$,   and  the private information state $\Xi_t^k$ is also needed $\forall (t,k)$. Similarly, we may use strategies   $\gamma_t^k=\gamma_t^k(\Xi_t^k, \Pi_t^{\gamma^{(K)}}, \Lambda_t^{k})$ or  $\gamma_t^k=\gamma_t^k(\Xi_t^k, \Pi_t^{\gamma^{(K)}}), \forall (t,k)$.

For  $\gamma_t^k=\gamma_t^k(\Xi_t^k, \Pi_t^{\gamma^{(K)}}, \Lambda_t^{k}), \forall (t,k)$, the single cost-to-go   (\ref{v-w-ctg}) satisfies,   $v_t^{\gamma^{(K),o}}(\delta_t^{(K)})={ v}_t^{\gamma^{(K),o}}(\pi_t)$, $\forall \delta_t^{(K)}, \forall \Pi_t^{\gamma^{(K),o}}=\pi_t$, and the following   DP equations: 
\begin{align}
&{ v}_n^{\gamma^{(K),o}}(\pi_n )
= \inf_{\gamma_n^{(K)}}\int_{{\mathbb X}_n  \times {\mathbb L}_{n}^{(K)} }     \ell(n, x_n, \gamma_n^{1}(\xi_n^{1}[\delta_n^{(K)}, \lambda_n^1], \pi_n, \lambda_n^{1}), \nonumber \\
&\ldots, \gamma_n^{K}(\xi_n^{K}[\delta_n^{(K)}, \lambda_n^K], \pi_n, \lambda_n^{K})) \pi_n (dx_{n}, d\lambda_n^{(K)}),      \label{com-dp-1} \\
&{ v}_t^{\gamma^{(K),o}}(\pi_t)   \hst \hso     \forall t \in T_+^{n-1} \nonumber
   \\
 &=  \inf_{\gamma_t^{(K)} }  \Big\{  \int_{{\mathbb X}_t \times {\mathbb L}_{t}^{(K)}}    \ell(t, x_t,  \gamma_t^{1}(\xi_t^{1}[\delta_t^{(K)}, \lambda_t^1], \pi_t, \lambda_t^{1}), \ldots  \nonumber \\
 &\gamma_t^{K}(\xi_t^{K}[\delta_t^{(K)},\lambda_t^K], \pi_t, \lambda_t^{K})).  \pi_{t}(dx_{t}, d\lambda_{t}^{(K)}) \nonumber \\
 & + {\mathbb E}^{\gamma^{(K)}}\Big\{ {v}_{t+1}^{\gamma^{(K),o}}(\Pi_{t+1}^{\gamma^{(K),o}})\Big|\Pi_{t}^{\gamma^{(K),o}}=\pi_{t} \Big\}  \label{com-dp-2}
\end{align}
where $ \Pi_{t}^{\gamma^{(K),o}}, \forall t$ satisfies the recursion of Theorem~\ref{thm:is-cen}.(2).\\
By  (\ref{com-dp-1}), (\ref{com-dp-2}), it is easy to verify that  the infimum occurs in the subset, $\gamma_t^{k,o}=\gamma_t^{k, is}(\xi_t^k, \pi_t), \forall (t,k)$.

\section{Example}
\label{sect:ex}

We illustrate the implementation of the  DP equations of  Theorem~\ref{thm:vp} to an  example, with  two agents  and two  observations $K=2$, spaces  ${\mathbb X}_t=\{s_1,s_2\}$, ${\mathbb Y}_t^k=\{o_1, o_2\}$, ${\mathbb   U}_t^k=\{c_1,c_2\},  t=1,2,3, k=1,2$, and   $2-$step delayed sharing information structures, i.e., $T=2$,  $I_t^k=\Delta_t^{(2)}\cup \Lambda_t^k$,
\begin{align*}
& \Delta_t^{(2)}=\{Y_{1,t-2}^{(2)},U_{1,t-2}^{(2)}\}, 
\Lambda_t^k=\{Y_{t-1,t}^k,U_{t-1}^k\}, \\
&I_1^k=\{Y_1^k\},
I_2^k=\{Y_{1,2}^k,U_1^k\},
I_3^k=\{\Delta_3^{(2)},Y_{2,3}^k,U_2^k\}
\end{align*}
where $t=1,2,3,  k=1,2$, 
with  $\Delta_t^{(2)}=\emptyset$ for $t<3$.  The  stage cost is time-invariant and  additive 
$\ell(t, x,u^1,u^2)=\ell^1(x,u^1)+\ell^2(x,u^2), \forall t$, i.e.,  separable. 
The problem data are given in Table~\ref{tab:obs-kernels},   Table~\ref{tab:S-kernel},  and 
Table~\ref{tab:l-kernel}. The  initial probabilities are  $ S(X_1=s_1)=0.7,  S(X_1=s_2)=0.3.$ 


The example is intended to represent a master--apprentice system. For this reason, the numerical values are chosen so that agent 1 has a more reliable observation channel $Q^1(\cdot|\cdot)$ than the channel $Q^2(\cdot|\cdot)$ of agent 2.

\begin{table}[t]
\caption{Observations transition probabilities}
\label{tab:obs-kernels}
\centering
\begin{tabular}{c|cc}
\hline
 & $o_1$ & $o_2$ \\
\hline
$Q^1(\cdot|s_1)$ & $0.9$ & $0.1$ \\
$Q^1(\cdot|s_2)$ & $0.1$ & $0.9$ \\
$Q^2(\cdot|s_1)$ & $0.7$ & $0.3$ \\
$Q^2(\cdot|s_2)$ & $0.3$ & $0.7$ \\
\hline
\end{tabular}
\end{table}

\begin{table}[t]
\caption{State transition probabilities  $S(\cdot|x,u^1,u^2)$}
\label{tab:S-kernel}
\centering
\begin{tabular}{c|c|cc}
\hline
$x$ & $(u^1,u^2)$ & $S(s_1|\cdot)$ & $S(s_2|\cdot)$ \\
\hline
$s_1$ & $(c_1,c_1)$ & $0.9$ & $0.1$ \\
$s_1$ & $(c_1,c_2)$ & $0.7$ & $0.3$ \\
$s_1$ & $(c_2,c_1)$ & $0.8$ & $0.2$ \\
$s_1$ & $(c_2,c_2)$ & $0.7$ & $0.3$ \\
$s_2$ & $(c_1,c_1)$ & $0.2$ & $0.8$ \\
$s_2$ & $(c_1,c_2)$ & $0.5$ & $0.5$ \\
$s_2$ & $(c_2,c_1)$ & $0.7$ & $0.3$ \\
$s_2$ & $(c_2,c_2)$ & $0.8$ & $0.2$ \\
\hline
\end{tabular}
\end{table}

\begin{table}[t]
\caption{Stage cost $\ell(x,u^1, u^2)=\ell^1(x, u^1)+ \ell^2(x,u^2)$}
\label{tab:l-kernel}
\centering
\begin{tabular}{c|cc|cc}
\hline
 & \multicolumn{2}{c|}{$\ell^1(x,u^1)$} 
 & \multicolumn{2}{c}{$\ell^2(x,u^2)$} \\
\hline
$x$ & $u^1=c_1$ & $u^1=c_2$ & $u^2=c_1$ & $u^2=c_2$ \\
\hline
$s_1$ & $0$ & $1$ & $0$ & {$1$} \\
$s_2$ & {$8$} & {$2$} & {$4$} & {$2$} \\
\hline
\end{tabular}
\end{table}

%

\begin{algorithm}[t]
\caption{Time-first PbP update}
\label{alg:pbp-example}
\begin{algorithmic}
\STATE Initialize an admissible strategy tuple $\gamma^{(K)}$
\REPEAT
    \STATE Store the current strategy tuple as $\gamma_{\rm old}^{(K)}$
    \FOR{$t=n$ downto $1$}
        \REPEAT
            \FOR{$k=1$ to $K$}
                \STATE Fix $\gamma^{-k}$
                \FOR{each admissible realization $i_t^k \in {\mathbb I}_t^k$}
                    \STATE Compute the private posterior recursively
                    \STATE Enforce  consistency of hidden actions in $\Lambda_t^{-k}$
                    \FOR{each action $u_t^k \in {\mathbb U}_t^k$}
                        \STATE Evaluate the right-hand side of Theorem~\ref{thm:vp}
                    \ENDFOR
                    \STATE Update $\gamma_t^k(i_t^k)$ using a minimizing action
                \ENDFOR
            \ENDFOR
        \UNTIL{no strategy at time $t$ changes}
    \ENDFOR
\UNTIL{$\gamma^{(K)}=\gamma_{\rm old}^{(K)}$}
\STATE Return the converged PbP strategy profile
\end{algorithmic}
\end{algorithm}

Algorithm~\ref{alg:pbp-example}   implements the DP equations of Theorem~\ref{thm:vp}. At each time step, each controller updates its policy while the other controller's policy is fixed, and this is done until there is no change on any controller. The private posterior needed for the DP equations is calculated recursively by taking the information realization as the argument of the value function. The hidden actions inside $\Lambda_t^{-k}$ should be consistent with the other controller's policy

\begin{table}[t]
\centering
\caption{Sample computed values for the finite example. The realizations follow  the order used in the definitions of $I_t^k$.}
\label{tab:sample-values}
\begin{tabular}{c c c c c}
\hline
Agent $k$  & Stage $t$  & Realization $I_t^k$ & Value Funct. & $u_t^{k,o}$ \\
\hline
$1$ & $t=1$ & $(o_2)$ & $6.71$ & $c_2$ \\
$1$ & $t=2$ & $(o_2,o_1,c_2)$ & $1.61$ & $c_1$ \\
$1$ & $t=3$ & $(o_2,o_2,c_2,c_2,o_1,o_1,c_1)$ & $0.33$ & $c_1$ \\
\hline
$2$ & $t=1$ & $(o_2)$ & $5.81$ & $c_2$ \\
$2$ & $t=2$ & $(o_2,o_1,c_2)$ & $1.99$ & $c_1$ \\
$2$ & $t=3$ & $(o_2,o_2,c_2,c_2,o_1,o_1,c_1)$ & $0.51$ & $c_1$ \\
\hline
\end{tabular}
\end{table}

Algorithm~\ref{alg:pbp-example} is first tested on a fully separated example, with state $X_t=(X_t^1,X_t^2)$, separated transition probabilities 
$S(x_{t+1}|x_t,u_t^{(2)})=\prod_{k=1}^2 S^k(x_{t+1}^k|x_t^k,u_t^k)$, separated observations 
$Q^{(2)}(y_t^{(2)}|x_t)=\prod_{k=1}^2 Q^k(y_t^k|x_t^k)$, and separated payoff 
$\ell(t,x,u^{(2)})=\ell^1(t,x^1,u^1)+\ell^2(t,x^2,u^2)$. For this fully separated problem, we deducec that  $J_n(\gamma^{(2),o})=J_n^1(\gamma^{1,o})+J_n^2(\gamma^{2,o})$, where $J_n^k(\gamma^{k,o})$ is the solution of  the  centralized POMDP corresponding to $S^k(x_{t+1}^k|x_t^k, u_t^{k}), Q^k(y_{t}^k|x_t^k), \ell^k(t, x^k, u^k) $ for $k=1,2$. This solution was also confirmed  by running  Algorithm~\ref{alg:pbp-example}.

Algorithm~\ref{alg:pbp-example} is applied to the example shown in 
Table~\ref{tab:obs-kernels}, Table~\ref{tab:S-kernel}, and 
Table~\ref{tab:l-kernel}, and it gives the PbP payoff
\begin{align*}
J_3(\gamma^{1,o},\gamma^{2,o})
&=
{\mathbb E}^{\gamma^{2,o}}[{\cal V}_1^{\gamma^{2,o}}(I_1^1)]=
{\mathbb E}^{\gamma^{1,o}}[{\cal V}_1^{\gamma^{1,o}}(I_1^2)]= 4.3005
\end{align*}
i.e., these are identical as expected. 
This gives a consistency check for the recursive posterior and value 
computations. The values presented in Table~\ref{tab:sample-values} are 
the values ${\cal V}_t^{\gamma^{-k,o}}(i_t^k)$ for selected admissible 
realizations. Since observation $o_2$ is more likely in the risky state 
$s_2$, both controllers choose the protective action $c_2$ at the first 
time step in the displayed histories. In the later displayed histories, 
the additional observations are more consistent with the less risky state 
$s_1$, and the lower-cost action $c_1$ becomes minimizing.


%

\section{conclusion}
The main contribution of this paper is  the development of a  DP framework  for  POMDPs,    with multiple agents or controls  each    assigned  $T-$step delayed sharing  information patterns, based on  the concept of  decentralized sequential team equilibrium, i.e.,  person-by-person optimality of  static team theory. Each agent  makes use of a private information state  and a centralized information state shared by all  agents. Necessary  and sufficient conditions for decentralized team equilibrium   are derived and a  generalized separation theorem is shown.  Central to the new  DP approach is the extension of fundamental properties of classical centralized DP of POMDPs to  decentralized POMDPs. Each  agent's strategy   compresses the   assigned information pattern into a private information state and a  centralized information state,  which is common to all strategies. This compression structural property of optimal strategies is a refinement of  Witsenhausen's \cite[Assertion 8, pp.1562]{witsenhausen1971}, which is sub-optimal according to the counterexample  constructed  by Varaiya and Walrand  \cite{varaiya-walrand1978}.

\section{Appendix}

\subsection{Proof of Theorem~\ref{thm:is-pbp}}
\label{app:thm:is-pbp}
Recall (\ref{inf_1})-(\ref{ISG-2}), $\{Y_{1,t}^{(K)}, U_{1,t-1}^{(K)}\}= \{\Lambda_t^k, \Delta_t^{(K)}, \Lambda_t^{-k}\}=\{\Delta_t^{(K)}, \Lambda_t^{(K)}\}$. 
By Assumptions~\ref{ass-1},
  the conditional PM  of $(X_{t+1}, \Lambda_{t+1}^{-k})$ given  $I_{t+1}^k\equiv \{\Lambda_{t+1}^k, \Delta_{t+1}^{(K)}\}=\{\lambda_{t+1}^k, \delta_{t+1}^{(K)}\}$,    is determined by using  Bayes' theorem as follows. 
\begin{align}
&{\mathbb P}^{{\gamma^k, \gamma^{-k}}} \Big\{X_{t+1} \in dx_{t+1}, \Lambda_{t+1}^{-k} \in d\lambda_{t+1}^{-k} \Big| \lambda_{t+1}^{k}, \delta_{t+1}^{(K)}\Big\}\nonumber \\
&=\frac{{\mathbb P}^{{\gamma^k, \gamma^{-k}}} \Big\{X_{t+1} \in dx_{t+1}, \Lambda_{t+1}^{-k} \in d\lambda_{t+1}^{-k}, d\lambda_{t+1}^k, d \delta_{t+1}^{(K)} \Big\}}{ {\mathbb P}^{{\gamma^k, \gamma^{-k}}} \Big\{d\lambda_{t+1}^{k},  d\delta_{t+1}^{(K)}\Big\}    }, \label{gen-rec_2}\\
&{\mathbb P}^{{\gamma^k, \gamma^{-k}}} \Big\{X_{t+1} \in dx_{t+1}, \Lambda_{t+1}^{-k} \in d\lambda_{t+1}^{-k}, d\lambda_{t+1}^k, d\delta_{t+1}^{(K)} \Big\} \nonumber \\
&=\int_{{\mathbb X}_t}{\bf  P}_{t+1}^{{\gamma^k, \gamma^{-k}}}(dx_{tt+1},dx_t,  d\lambda_{t+1}^{-k}, d\lambda_{t+1}^{k}, d\delta_{t+1}^{(K)}).
\label{gen-rec_3}
\end{align}
 The PM under the integrand in    (\ref{gen-rec_3})   is determined by 
\begin{align}
&{\bf  P}_{t+1}^{{\gamma^k, \gamma^{-k}}}(dx_{t+1},dx_t,  d\lambda_{t+1}^{-k}, d\lambda_{t+1}^{k}, d\delta_{t+1}^{(K)}) \\
 &= Q_{t+1}^{(K)}(dy_{t+1}^{(K)}|x_{t+1}, u_{t}^{(K)}) S_{t+1}(dx_{t+1}|x_t, u_t^{(K)}) \nonumber \\
 &. {\bf  P}_t^{{\gamma^k, \gamma^{-k}}}(dx_t, dy_{t-T+2, t}^{-k} du_{t-T+2, t}^{-k},\nonumber \\
 & dy_{t-T+2, t}^k, du_{t-T+2, t}^k, dy_{1,t-T+1}^{(K)}, du_{1, t-T+1}^{(K)}) . \label{der_1}
\end{align}
Focusing on the last PM  in (\ref{der_1}),  we have
\begin{align}
& {\bf  P}^{{\gamma^k, \gamma^{-k}}}(dx_t, dy_{t-T+2, t}^{-k} du_{t-T+2, t}^{-k},\nonumber \\
 & dy_{t-T+2, t}^k, du_{t-T+2, t}^k, dy_{1,t-T+1}^{(K)}, du_{1, t-T+1}^{(K)}) \\
&= {\bf  P}_t^{{\gamma^k, \gamma^{-k}}}(du_{ t}^{-k}\big| x_t, \lambda_{t}^{-k}, \lambda_{t}^k, u_t^k, \delta_t^{(K)})\\
 &. {\bf  P}_t^{{\gamma^k, \gamma^{-k}}}(dx_t, d\lambda_t^{-k}\big| \lambda_t^k,  \delta_t^{(K)},u_t^k)  {\bf  P}_t^{{\gamma^k, \gamma^{-k}}}(d \lambda_t^k,  \delta_t^{(K)}, du_t^k) \nonumber    \\
 &\sr{(a)}{=} {\bf  P}_t^{{\gamma^k, \gamma^{-k}}}(du_{ t}^{-k}\big| \lambda_{t}^{-k}, \delta_t^{(K)})     {\bf  P}_t^{{\gamma^k, \gamma^{-k}}}(dx_t, d\lambda_t^{-k}\big| \lambda_t^k, \delta_t^{(K)})   \nonumber \\
 &.  {\bf  P}_t^{{\gamma^k, \gamma^{-k}}}(d \lambda_t^k,  \delta_t^{(K)}, du_t^k) \label{der_3}
\end{align}
where $(a)$ is due to the conditional independence condition (\ref{CI}), (\ref{MC}),  and $u_t^k=\gamma_t^k(\delta_t^{(K)}, \lambda_{t}^k)$,   which implies  ${\bf  P}_t^{{\gamma^k, \gamma^{-k}}}(dx_t, d\lambda_t^{-k}\big| \lambda_t^k, u_t^k, \delta_t^{(K)}) = {\bf  P}_t^{{\gamma^k, \gamma^{-k}}}(dx_t, d\lambda_t^{-k}\big| \lambda_t^k, \delta_t^{(K)})$, i.e.,    the control $u_t^k$ is redundant.  
Finally, by (\ref{CI}), (\ref{MC}) then  ${\bf  P}_t^{{\gamma^k, \gamma^{-k}}}(du_{ t}^{-k}\big| \lambda_{t}^{-k}, \delta_t^{(K)}) =\prod_{j=1, j \neq k}^K P_t^j(du_t^{j}\big| \delta_t^{(K)}, \lambda_t^j)  =\prod_{j=1, j \neq k}^K { \mu}_{\gamma_t^{j}(\delta_t^{(K)}, \lambda_t^j)}(du_t^{j})$.  
Substituting this  into   (\ref{der_3}), and (\ref{der_3}) into (\ref{der_1}), and the  resulting (\ref{der_1}) into (\ref{gen-rec_2}), we obtain
 (\ref{apost_1}).

\subsection{Proof of Theorem~\ref{thm:vp}}
\label{app;thm:vp}
 (1.1) By nested property,  (\ref{eq-nested}), 
 $I_t^k \subseteq I_{t+1}^k$, and 
  (\ref{opt-ctg-g}), 
\begin{align}
&{\cal V}_{t}^{\gamma^{k,o},\gamma^{-k,o}}(i_t^k) \nonumber \\
&=\inf_{\gamma_t^k \in {\cal U}_t^k}   {\mathbb E}^{^{\gamma^k, \gamma^{-k,o}}} \Big\{  \ell(t, X_t,\gamma_t^k,  \gamma_t^{-k,o})+\inf_{ \gamma_{t+1,n}^k\in {\cal U}_{t+1,n}^k}\Big( \nonumber \\
& {\mathbb E}^{^{\gamma^k, \gamma^{-k,o}}} \Big[ \sum_{j=t+1}^{n} \ell(j, X_j,\gamma_j^k,  \gamma_j^{-k,o}) \Big| I_{t+1}^k \Big]\Big) \Big| i_t^k \Big\} \label{opt-ctg-p-1}\\
&=\inf_{\gamma_t \in {\cal U}_t^k}   {\mathbb E}^{^{\gamma^k, \gamma^{-k,o}}} \Big\{  \ell(t, X_t,\gamma_t^k,  \gamma_t^{-k,o}) \nonumber \\
&\hst +{\cal V}_{t+1}^{\gamma^{k,o}, \gamma^{-k,o}}(I_{t+1}^k)   \Big| i_t^k, \gamma_t^k \Big\}\label{opt-ctg-p-2_a-1}
\end{align}
where  (\ref{opt-ctg-p-2_a-1}) is by definition of value process and    $i_t^k$  specifies $\gamma_t^k\equiv \gamma_t^k(i_t^k)=\gamma_t^k(\delta_t^{(K)}, \lambda_t^k)\in {\mathbb U}_t^k $.
 By Lemma~\ref{lemma:nested}, 
  (\ref{eq-nested-a-1}),
\begin{align}
  &{\mathbb E}^{^{\gamma^k, \gamma^{-k,o}}} \Big\{{\cal V}_{t+1}^{\gamma^{k,o}, \gamma^{-k,o}}(I_{t+1}^k) \Big| i_t^k, \gamma_t^k \Big\} \sr{(a)}{=}{\mathbb E}^{^{\gamma^k, \gamma^{-k,o}}} \Big\{ \nonumber \\ &{\cal V}_{t+1}^{\gamma^{k,o}, \gamma^{-k,o}}(i_{t}^k,Y_{t+1}^k, \gamma_t^k, Y_{t-T+1}^{-k}, U_{t-T+1}^{-k})  \Big| i_t^k, \gamma_t^k \Big\}\nonumber  \\
 &\sr{(b)}{=}  \int_{ {\mathbb Y}_{t+1}^k \times {\mathbb X}_{t, t+1}   \times {\mathbb L}_{t}^{-k}    } {\cal V}_{t+1}^{ \gamma^{k,o},  \gamma^{-k,o}}(i_{t}^k,  y_{t+1}^k, \gamma_t^k, y_{t-T+1}^{-k}, 
 u_{t-T+1}^{-k}) \nonumber \\
 &  .{\bf P}_{t+1}^{\gamma^{k},  \gamma^{-k,o}}(dy_{t+1}^k, d\lambda_{t}^{-k}, dx_{t}, dx_{t+1}\big|   i_t^k, \gamma_t^k)\Big\} \label{nm-ctg}
\end{align}
where in $(a)$,  by  Lemma~\ref{lemma:nested}.(2),   $U_{t-T+1}^{-k}=\gamma_{t-T+1}^{-k,o}=\gamma_{t-T+1}^{-k,o}(\Delta_{t-T+1}^{(K)}, \Lambda_{t-T+1}^{-k})$ and    the argument of   $\gamma_{t-T+1}^{-k,o}(\cdot)$, except  $Y_{t-T+1}^{-k}$ 
%
 is  specified by $\delta_{t}^{(K)}\in i_t^k$   as stated 
 below (\ref{dp_1_na-str1}), 
%
and $(b)$ is obtained by   
 expressing the conditional expectation w.r.t.   ${\bf P}_{t+1}^{\gamma^{k},  \gamma^{-k,o}}( dy_{t+1}^k, d\lambda_t^{-k},dx_{t}, dx_{t+1} \big|i_t^k,\gamma_t^k)$.
Applying Bayes' Theorem,  using the Markovian network,  (\ref{i-d1}),  (\ref{i-d2}),   $I_t^k=i_t^k\tri \{\delta_t^{(K)}, \lambda_t^k\}$,   $\lambda_t^{(K)}\tri \{\lambda_t^k, \lambda_t^{-k}\}$, and Lemma~\ref{lem:ind} we obtain (\ref{Mar_REC_1-nn-11})-(\ref{Mar_REC_1-nn}).
Next, we show the optimization is over actions and not strategies, as in classical DP \cite{kumar-varayia:B1986,bertsekas-shreve1978}.
  At  $t=n$, since only the terminal cost is present, 
 by using the  conditional independence and  Lemma~\ref{lem:ind}, 
\begin{align}
&{\cal V}_n^{\gamma^{k,o}, \gamma^{-k,o}}(i_n^k)
= \inf_{\gamma_n^k \in {\cal  U}_n^k}\int_{{\mathbb X}_n  \times {\mathbb L}_{n}^{-k} }     \ell(n, x_n,\gamma_n^k(i_n^k), \nonumber \\
& \gamma_n^{-k,o}(\delta_n^{(K)}, \lambda_n^{-k}))  \xi_n^k[i_n^k]  (dx_{n}, d\lambda_n^{-k})  . \label{dp_1_nnn_1-f}  
\end{align}
Similar to \cite{bertsekas-shreve1978}, in  (\ref{dp_1_nnn_1-f})
 the  infimum  over the  strategy  $\gamma_n^k(\cdot)  \in {\cal  U}_n^k$  is equivalent  to  infimum over actions  $\gamma_n^k(i_n^k) =U_n^k  \in {\mathbb U}_n^k$, hence  ${\cal V}_{n}^{\gamma^{k,o},\gamma^{-k,o}}(i_{n}^k) ={\cal V}_{n}^{\gamma^{-k,o}}(i_{n}^k  ), \forall \gamma^{o, k}$,   (\ref{dp_1_nnn}),   (\ref{dp_1_nnn_1}) are obtained. 
For any $t \in T_+^{n-1}$,   substituting (\ref{nm-ctg}) into (\ref{opt-ctg-p-2_a-1}), and using the fact that the 
 conditioning  information  is    $\{i_t^k, \gamma_t^k(i_t^k)\}, U_t^k=\gamma_t^k\in {\mathbb U}_t^k$, then the  infimum is over the action spaces and not the strategy spaces. Similarly,  we can also show,    ${\cal V}_{t+1}^{\gamma^{k,o},\gamma^{-k,o}}(i_{t+1}^k) ={\cal V}_{t+1}^{\gamma^{-k,o}}(i_{t+1}^k  ), \forall \gamma^{o, k}$, and 
 we obtain (\ref{dp_1_na}) and also  (\ref{dp_2_nm}). (1.2)  At time $t=n$, by  (\ref{dp_1_nnn_1}), we  deduce  the  optimal is $u_n^k=u_n^{k,o}=\gamma_n^{k,o}(\xi_n^k, \delta_n^{(K)}) \in {\mathbb U}_n^k$, hence  $\gamma_n^{k,o}(\cdot) \in  {\cal U}_{n}^{k}$.  For any $ t\in  T_+^{n-1}$, by   (\ref{dp_1_na}) and (\ref{Mar_REC_1-nn})  we deduce the optimal is $u_t^k=u_t^{k,o}=\gamma_t^{k,o}(\xi_t^k, \delta_t^{(K)}, \lambda_t^k)\in {\mathbb U}_t^k,$  $i_t^k=\{\delta_t^{(K)}, \lambda_t^k\}$,  $\gamma_t^{k,o}(\cdot) \in  {\cal U}_{t}^{k}$, Since $\xi_t^k\equiv \xi_t^k[i_t^k]$ is measurable w.r.t. $i_t^k$ then  $\gamma^{k,o}(\cdot) \in  {\cal U}_{1,n}^{k}$. 
  (2) This  follows by  induction using the DP equations of (1); we omit the proof  because it is a repetition  of Theorem~\ref{thm:vf-g}.(2).

\bibliographystyle{IEEEtran}

\bibliography{Bibliography_Decentralized_arxiv_FIXED_2026-05-26}

\end{document}